\documentclass[draftcls,onecolumn,12pt]{IEEEtran}
\usepackage[width=6.5in, left=2.5cm,height=54em]{geometry}

\usepackage[utf8]{inputenc}
\usepackage{framed}
\usepackage[linesnumbered,lined, algoruled]{algorithm2e}
\usepackage{algorithmic,float}
\usepackage{amsmath}
\usepackage{amsmath}
\usepackage{amsthm}
\usepackage{amsfonts}
\usepackage{amssymb}
\usepackage{bm,array}
\usepackage{soul,color}
\usepackage[acronym,shortcuts]{glossaries}
\usepackage{graphicx}
\usepackage{graphics}
\usepackage{comment}
\usepackage{setspace}
\usepackage[inline]{enumitem}
\makeglossaries
\DeclareMathOperator\erf{erf}
\newacronym{ae}{AE}{auto encoder}
\newacronym{auc}{AUC}{area under the curve}
\newacronym{ap}{AP}{access point}
\newacronym{ce}{CE}{cross entropy}
\newacronym{cdf}{CDF}{cumulative distribution function}
\newacronym{eda}{EDA}{estimated distance approach}
\newacronym{fa}{FA}{false alarm}
\newacronym{gnss}{GNSS}{global navigation satellite system}
\newacronym{irlv}{IRLV}{in-region location verification}
\newacronym{kl}{K-L}{Kullback-Leibler}
\newacronym{ls}{LS}{least-squares}
\newacronym{llr}{LLR}{log likelihood-ratio}
\newacronym{los}{LOS}{line of sight}
\newacronym{glrt}{GLRT}{generalized likelihood ratio test}
\newacronym{gsm}{GSM}{global system for mobile communications}
\newacronym{lssvm}{LS-SVM}{least squares SVM}
\newacronym{md}{MD}{mis-detection}
\newacronym{ml}{ML}{machine learning}
\newacronym{mlp}{MLP}{multi-layer perceptron}
\newacronym{mmse}{MMSE}{minimum mean square error}
\newacronym{mse}{MSE}{mean squared error}
\newacronym[\glslongpluralkey={neural networks}]{nn}{NN}{neural network}
\newacronym{np}{N-P}{Neyman-Pearson}
\newacronym{oclssvm}{OCLSSVM}{one-class least-square \ac{svm}}
\newacronym{pdf}{PDF}{probability density function}
\newacronym{pso}{PSO}{particle swarm optimization}
\newacronym{roc}{DET}{detection error tradeoff}
\newacronym{det}{DET}{detection error tradeoff}
\newacronym{roi}{ROI}{region of interest}
\newacronym{rss}{RSS}{received signal strength}
\newacronym{std}{std}{standard deviation}
\newacronym[\glslongpluralkey={support vector machines}]{svm}{SVM}{support vector machine}
\newacronym{ue}{UD}{user device}
\expandafter\def\expandafter\quote\expandafter{\quote\setstretch{1}~\\}
\usepackage[normalem]{ulem}
\usepackage[user]{zref}
\newcommand{\sot}[1]{}

\newcounter{revc}
\makeatletter \zref@newprop{revcontent}{} \zref@addprop{main}{revcontent}
\zref@newprop{revsec}{} \zref@addprop{main}{revsec}

\newcommand{\revi}[2]{%
\zref@setcurrent{revsec}{\thesection}%
\zref@setcurrent{revcontent}{#2}%
\refstepcounter{revc}%
\label{#1}
\zlabel{#1}%
#2 }

\newcommand{\revinu}[2]{%
\zref@setcurrent{revsec}{\thesection}%
\zref@setcurrent{revcontent}{#2}%
\refstepcounter{revc}%
\zlabel{#1}%
\label{#1}
#2 }

\newcommand{\revr}[2]{%
\zref@setcurrent{revsec}{\thesection}%
\zref@setcurrent{revcontent}{#2}%
\refstepcounter{revc}%
\zlabel{#1}%
\label{#1} \sot{#2}} \makeatother

\newcommand{\ie}{i.e., }
\newcommand{\wrt}{with respect to }
\newcommand{\Exp}[1]{\mathbb{E}\left[#1\right]}
\newcommand{\ai}{\bm{a}^{(i)}}

\DeclareMathOperator{\sign}{sign}
\newcommand{\E}{E}
\newtheorem{theorem}{Theorem}
\newtheorem{lemma}{Lemma}
\newtheorem{corollary}{Corollary}

\title{Machine Learning For In-Region Location Verification In Wireless Networks}
\author{Alessandro Brighente, Francesco Formaggio,\\  Giorgio Maria Di Nunzio, and  Stefano Tomasin \\
{\bf Accepted for publication in the IEEE Journal on Selected Areas in Communications (JSAC), May 2019.}
 }

\begin{document}

\maketitle

\sloppy

\begin{abstract}
\footnote{This work has been submitted to the IEEE for possible publication. Copyright may be transferred without notice, after which this version may no longer be accessible}\Ac{irlv} aims at verifying whether a user is inside a \ac{roi}. In wireless networks, \ac{irlv} can exploit the features of the channel between the user and a set of trusted access points. In practice, the channel feature statistics is not available and we resort to \ac{ml} solutions for \ac{irlv}. We first show that solutions based on either \acp{nn} or \acp{svm} and typical loss functions are \ac{np}-optimal at learning convergence for sufficiently complex learning machines and large training datasets . Indeed, for finite training, \ac{ml} solutions are more accurate than the \ac{np} test based on estimated channel statistics. Then, as estimating channel features outside the \ac{roi} may be difficult, we consider one-class classifiers, namely auto-encoders \acp{nn} and one-class \acp{svm}, which however are not equivalent to the \acf{glrt},  typically replacing the \ac{np} test in the one-class problem.  Numerical results support the results in realistic wireless networks, with channel models including path-loss, shadowing, and fading.
\end{abstract}


\glsresetall
\clearpage
\section{Introduction}
\label{sec:intro}

Location information without verification gives ample opportunities to attack a service granting system (with  applications in sensor networks \mbox{\cite{Zeng-survey, 8376254, wei2013}}, the Internet of things (IoT) \cite{7903611}, and geo-specific encryption \cite{quaglia}). In fact, the location information can be easily manipulated either by tampering the hardware/software reporting the location or by spoofing the \ac{gnss} signal outside the user device. In this context, location verification systems aim at verifying the position of mobile devices in a communication network.  In order to verify the location, the features of the wireless channel over which communications occur can be exploited. An example is given by {\cite{li2010security}}, where the \ac{rss} is used to estimate the distance between the user and other network nodes.

\revi{rev2lit}{Location verification can be classified into two main sub-problems: \emph{single location verification} and {\it \ac{irlv}}. The \emph{single location verification} problem aims at verifying if a user is in a specific point. A solution is obtained by comparing some channel features of the user under test with those of a trusted user that was in the same location in the past. }\revi{rev3cit}{In some works, this approach is used to verify if different messages come from the same user, i.e., as a {\it user authentication} mechanism (see \protect{\cite{7270404}} for a survey): in \protect{\cite{Baracca-12}}, channel features are affected by noise with known  statistics; whereas, in \protect{\cite{7398138}}, statistics are unknown and a learning strategy is adopted.  The {\it \ac{irlv}} aims at verifying if a user is inside a \ac{roi} \cite{Zeng-survey}. }\revi{rev2lit2}{Solutions include distance bounding techniques with rapid exchanges of packets between the verifier and the prover \cite{Brands, singelee2005location}, also in the context of vehicular ad-hoc networks {\cite{song2008secure}}. Other solutions use radio-frequency and ultrasound signals {\cite{Sastry}}, or anchor nodes and transmit power variations {\cite{Vora}}. More recently, a delay-based verification technique leveraging geometric properties  has been proposed in {\cite{7145434}}. Some of the proposed techniques partially neglect wireless propagation phenomena (such as shadowing and fading) that corrupt the distance estimates {\cite{Brands,Sastry} and \cite{Vora}}. Other approaches assume specific channel statistics that may be not accurate due to changing environment conditions {\cite{quaglia}}.} 
\revi{attack1}{Two types of attacks to \ac{irlv} have been considered in the literature, where the attacker claims a false location \mbox{\cite{singelee2005location,song2008secure,Sastry}}  or tampers with the signal power making it coherent with the fake claimed position  \mbox{\cite{Vora,yan2016location}} and \mbox{\cite{li2010security}}.}

Focusing on \ac{irlv}, if the statistics of the channels to devices both inside and outside the \ac{roi} is known to the network, the \ac{np} theorem {\cite{Cover-book}} provides the most powerful test for a given significance level. When  the channel statistics is not available, a two-step solution would be to a) estimate the channel statistics and b) apply the \ac{np} theorem on the estimated statistics. However, as we also confirm in this paper, this approach may not be accurate. Alternatively, \ac{ml} techniques can be used. For example, in \cite{xiao-2018}, the single location verification problem is solved without assumptions  on the channel model by applying logistic regression. In \cite{tian2015robust}, the objective is to determine the position of a user inside a building by means of a multi-class \ac{svm}. Nevertheless, neither \cite{xiao-2018}  nor \cite{tian2015robust} compare the performance of their \ac{ml} approaches with that of the \ac{np} test.

In this paper, we remove the channel knowledge assumption and study two \ac{ml} solutions for \ac{irlv} based on \acp{nn} and \ac{svm}. In particular, we investigate \acp{mlp} that use either the \ac{ce} or the \ac{mse} as loss functions, and the \ac{ls} version of \acp{svm}. We show that these approaches are \ac{np}-optimal for sufficiently complex machines and sufficiently large training datasets. \revi{rev11a}{The obtained asymptotic results are applicable also to elaborate ML solutions, such as deep learning NNs, that can still be seen as parametric functions, although more complex than shallow NNs.}

Since it may be difficult to obtain training data from the space outside the \ac{roi}, as it can be vast or not well defined, we explore the one-class classification problem under the knowledge of legitimate channel statistics, and we conclude that conventional \ac{ml} solutions based on both the \ac{ae} and the one-class \ac{svm} do not coincide with the \ac{glrt}, even for large training datasets.  
Numerical results support the theoretical results in a realistic wireless network scenario, including path-loss, shadowing, and fading. We show that in a simple scenario a shallow \ac{nn} and a relatively small training dataset already provide optimal performance. We also show that one-class \ac{irlv}  achieves a performance comparable to that of two-class \ac{irlv}.

The contributions of this paper are summarized hereby:
\begin{enumerate}
    \item we propose physical-layer \ac{irlv} solutions based on \ac{ml} techniques that are suitable to operate with inaccurate estimates, even when their statistics are not known, thus being model-less;
    \item we show that, in asymptotic training and complexity conditions, \ac{nn} and \ac{lssvm} at convergence achieve the error probabilities of the \ac{np} test, which is most powerful for a given significance level.
\end{enumerate}
\revi{lit2}{About point 1, shadowing and fading effects on \ac{irlv} have not been much considered in the literature: for example, in \protect{\cite{Vora},} \ac{rss} estimates are assumed to be perfect; in {\cite{Sastry}}, agents are assumed to communicate over an  error-free channel (san assumption used for most distance-bounding protocols {\cite{singelee2005location,song2008secure}}). In {\cite{li2010security}} and {\cite{yan2016location}}, shadowing is taken into account, while fading is neglected, and channel statistics is assumed to be known. All these simplifying assumptions are not required by the \ac{ml} models studied in this paper.} \revi{framework1}{Indeed, we also consider an accurate wireless channel model (in Section II), but only in order to explain the complexity of the techniques in the literature (including the \ac{np} test) and, consequently, justify the use of \ac{ml}. Still,  our solution and theoretical results can be applied on any channel statistics and various features (see \cite{Zeng-survey} for a survey), even including  measurements from external sensors.}

The paper is organized as follows: Section~II introduces the system model for the \ac{irlv} problem, with an example of wireless channel, and recall two reference \ac{irlv} techniques. Section~III describes the proposed \ac{ml} solutions and presents the theoretical results on their asymptotic performance. In Section~IV, we propose the one-class classification approaches. Numerical results are shown and discussed in Section~V. Conclusions are outlined in Section~VII.

The following notation is used throughout the paper: bold lowercase letters refer to vectors, whereas bold uppercase letters refer to matrices, $\mathbb{E}[\cdot]$ and $\mathbb P[\cdot]$ denote the expectation and probability operators, respectively, $(\cdot)^T$ denotes the transpose operator, $\ln x$, and $\log_{10} x$ denote the natural-base and base-10 logarithms, respectively.

\section{System Model}


We consider a wireless network  with $N_{\rm AP}$ \acp{ap} covering the area $\mathcal{A}$ over a plane. We propose a \ac{irlv} system to determine if a \ac{ue} is transmitting from within an {\em authorized} \ac{roi} $\mathcal{A}_0$ inside  $\mathcal{A}$, and we define $\mathcal{A}_1=\mathcal{A} \setminus \mathcal{A}_0$ as its complementary region. The verification process exploits the location dependency of the features of the channel between the \ac{ue} and the \acp{ap}. \revi{revPHASE}{For example, we consider as feature the channel power attenuation (of a narrowband transmission), similarly to \protect{\cite{li2010security,Vora}} and \protect{\cite{yan2016location}}. Indeed, other features can be exploited, such as the phase or the wideband impulse response (see \protect{\cite{7270404}} for a survey): our solutions readily apply also to these cases, as we do not make special assumptions on the channel model for their design and analysis.}

We assume that the \ac{ue} transmits a pilot signal with fixed power, known at the \acp{ap}, from which the \acp{ap} can measure the received power and estimate the channel attenuation. We assume that the attenuation estimation is perfect, i.e., not affected by noise or interference, thanks to a sufficiently long pilot signal.

\subsection{Channel Model}\label{sec:chMod}

We now describe a widely adopted wireless channel model to clarify the challenge faced by an \ac{irlv} based on the attenuation estimate. \revi{WiFi2}{In particular, we consider the general channel \cite{3gpp} model that covers a large frequency range from $800$~MHz to 2.5~GHz, suitable for wireless local area networks (WLANs) and IoT, where \ac{irlv} is typically applied.} Let $a^{(n)}$ be the attenuation incurred over the channel between the \ac{ue} and \ac{ap} $n$, including the effects of path-loss, shadowing, and fading. In particular, by assuming a Rayleigh model for fading we have \begin{equation}
    g^{(n)} = (\sqrt{a^{(n)}})^{-1} \sim {\mathcal N}\left(0,\sigma_{g,n}^2\right),
\end{equation}
where ${\mathcal N}(m,\sigma^2)$ denotes a Gaussian random variable with mean $m$ and variance $\sigma^2$. Moreover, due to shadowing we have 
\begin{equation}
    (\sigma_{a,n}^2)_{\rm dB} =  -10\log_{10}\sigma_{g,n}^2={P_{\rm PL}^{(n)}} + s,
\end{equation}
where $P_{\rm PL}^{(n)}$ is the path-loss coefficient in dB, and $s \sim \mathcal{N}(0,\sigma_{s,{\rm dB}}^2)$ is the shadowing component.   Shadowing components of two \acp{ue}  at  positions $\bm{x}_1$ and $\bm{x}_2$   have correlation $\sigma_{s,{\rm dB}}^2e^{-\frac{L(\bm{x}_1,\bm{x}_2)}{d_c}}$, where $d_c$ is the shadowing decorrelation distance {\cite[Section 2.7]{goldsmith2005}.} 

Let us denote as $\bm{x}_{\rm AP}^{(n)} =(X_{\rm AP}^{(n)},Y_{\rm AP}^{(n)})$ the position of  \ac{ap} $n= 1, \ldots, N_{\rm AP}$. For a \ac{ue} located at $\bm{x}_{\rm UD}=(X_u,Y_u)$, its distance from \ac{ap} $n$ is $L(\bm{x}_{\rm UD},\bm{x}_{\rm AP}^{(n)}) = \sqrt{||\bm{x}_{\rm UD}-\bm{x}_{\rm AP}^{(n)}||_2^2}$. For the path-loss, \cite{3gpp} provides two scenarios: \ac{los} and non-\ac{los}. For a \ac{los} link, the path-loss in dB is modelled as
\begin{equation}\label{eq:los}
    P_{{\rm PL},{\rm LOS}}\left(L(\bm{x}_{\rm UD},\bm{x}_{\rm AP}^{(n)})\right) = 10 \nu \log_{10}\left(\frac{f 4\pi L(\bm{x}_{\rm UD},\bm{x}_{\rm AP}^{(n)})}{c}\right),
\end{equation}
where $\nu$ is the path-loss coefficient, $f$ is the carrier frequency and $c$ is the speed of light. 
For a  non-\ac{los} link, the path-loss coefficient in dB is defined as
\begin{equation}\label{eq:nlos}
\begin{split}
    P_{{\rm PL},  {\rm NLOS}}&\left(L(\bm{x}_{\rm UD},\bm{x}_{\rm AP}^{(n)})\right) = 40 (1 - 4 \cdot 10^{-3} h_ {\rm AP}^{(n)})\log_{10}\left (\frac{L(\bm{x}_{\rm UD},\bm{x}_{\rm AP}^{(n)})}{10^3}\right ) +   \\
    & - 18 \log_{10} h_{\rm AP}^{(n)}  +21\log_{10}\left(\frac{f}{10^6}\right) + 80,
    \end{split}
\end{equation}
where $h_{\rm AP}$ is the AP height. Path-loss and shadowing components (thus $\sigma_{a,n}^2$) are assumed to be time-invariant, while the fading (thus attenuation $a^{(n)}$) is independent for each attenuation estimate. \revi{avg_1}{Fading does not give information on the \ac{ue} location; therefore it is a disturbance for \ac{irlv}. However, by performing $k_f$ estimates of the attenuation in a short time $a_j^{(n)}$, $j=1, \ldots,k_f$, and averaging them, we obtain the new attenuation estimate}
\begin{equation}
    a^{(n)}_{\Sigma} = \frac{1}{k_f}\sum_{j=1}^{k_f} a_j^{(n)}.
\end{equation}

\subsection{\ac{irlv} With Known Channel Statistics}\label{sec:auth}

\ac{irlv} can be seen as a hypothesis testing problem between the two hypotheses (events):
\begin{itemize}
    \item $\mathcal{H}_0$: the \ac{ue} is transmitting from area $\mathcal{A}_0$;
    \item $\mathcal{H}_1$: the \ac{ue} is transmitting from area $\mathcal{A}_1$.
\end{itemize}
This is also denoted as a two-class classification problem. Given vector $\bm{a} = [a^{(1)}, \ldots, a^{(N_{\rm AP})}]$ collecting the attenuation estimates at all the \acp{ap}, we aim  at determining the most likely hypothesis, in order to perform \ac{irlv}. Let $\mathcal H \in  \{\mathcal{H}_0, \mathcal{H}_1\}$ be the state of the \ac{ue}, and $\hat{\mathcal H} \in  \{\mathcal{H}_0, \mathcal{H}_1\}$ the decision taken by the \acp{ap}. We have two possible errors: \acp{fa}, which occur when the \ac{ue}  is classified as outside the \ac{roi}, while being inside it, and \acp{md}, which occur when the \ac{ue}  is classified as inside the \ac{roi}, while being outside of it. We indicate the \ac{fa} probability as $P_{\rm FA} =\mathbb{P}(\hat{\mathcal H} = \mathcal H_1 | \mathcal H = \mathcal H_0)$, and the \ac{md} probability as $P_{\rm MD}=\mathbb{P}(\hat{\mathcal H} = \mathcal H_0 | \mathcal H = \mathcal H_1)$. 
Let $p(\bm{a}|\mathcal{H}_i)$ be the \ac{pdf} of observing the vector $\bm{a}$ given that $\mathcal{H} = \mathcal{H}_i$. The \ac{llr} for the considered hypothesis is defined as 
\begin{equation}\label{eq:lr}
    {\mathcal M}(\bm{a})=\ln \frac{p(\bm{a}|\mathcal{H}_0)}{p(\bm{a}|\mathcal{H}_1)}.
\end{equation}
According to the \ac{np} theorem, the most powerful test is obtained by comparing $\mathcal{M}(\bm{a})$ with a threshold value $\Lambda$, i.e., obtaining the test function
\begin{equation}
\label{eq:oneClassDec}
f^*(\bm{a}) =
\begin{cases}
-1 &\text{if } {\mathcal M}(\bm{a}) \geq \Lambda, \\
1 & \text{if } {\mathcal M}(\bm{a}) < \Lambda,
\end{cases}
\end{equation}
where $f^*(\bm{a})=-1$ corresponds to $\hat{\mathcal{H}}=\mathcal{H}_0$ and $f^*(\bm{a})=1$ corresponds to $\hat{\mathcal{H}}=\mathcal{H}_1$.
\revi{lambda}{Parameter $\Lambda$ must be chosen to obtain a desired significance level, \ie a desired \ac{fa} probability. It can be set either by assessing the \ac{fa} probability through simulations or by inverting, when available, the expression of \ac{fa} probability as a function of $\Lambda$ \protect{\cite[Section 3.3]{Kay-book}}.}

\subsection{Example of \ac{np} Test}\label{sec:los}
\begin{figure}
    \centering
    \includegraphics[width=6cm]{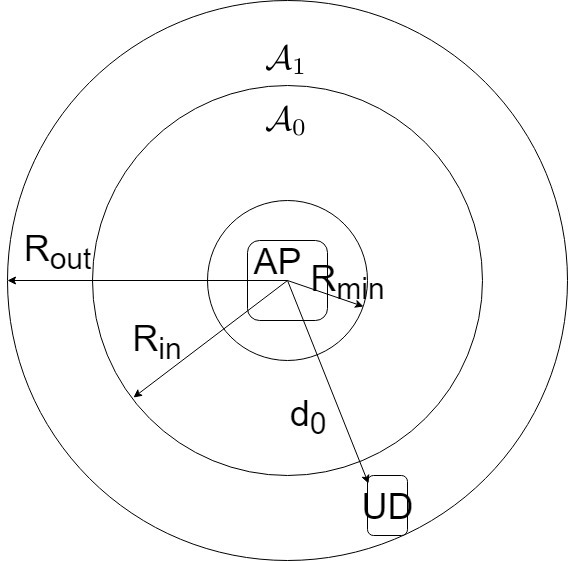}
    \caption{Simplified scenario with a single \ac{ap} located at the center of a circular \ac{roi}.}
    \label{fig:simpScen}
\end{figure}
\revi{simpleScen}{As example of application of the \ac{np} test we consider the scenario of Fig. \ref{fig:simpScen}, wherein area $\mathcal{A}$ is a ring with smaller radius $R_{\rm min}$ and larger radius $R_{\rm out}$ and \ac{roi} $\mathcal{A}_{0}$ is a ring concentric to $\mathcal{A}$, with larger radius $R_{\rm in}$ and smaller radius $R_{\rm min}$. A single \ac{ap} ($N_{\rm AP} =1$) is located at the \ac{roi} center and a \ac{ue} is transmitting from distance $d_0$. We consider two models: a) \emph{uncorrelated fading scenario}, which includes \ac{los} path-loss and spatially uncorrelated fading, and b) \emph{uncorrelated shadowing scenario}, which includes \ac{los} path loss and spatially uncorrelated shadowing. In both case, we consider the \ac{los} model for path-loss.}
\revi{simpleScen2}{In order to compute}\footnote{Note that for a single \ac{ap} vector $\bm{a}$ becomes the scalar $a$.} \revi{simpleScen2_0}{$p(a|\mathcal{H}_i)$, we first observe that the \ac{pdf} of incurring in attenuation $a$ for a user located inside the \ac{roi} is (by the total probability law)}
\begin{equation}\label{eq:prc}
p(a|\mathcal{H}_0) = \int_{R_{\rm min}}^{R_{\rm in}} p\left( a | d_0\right)p(d_0| d_0 \in \mathcal{A}_0) \, {\tt d}d_0,    
\end{equation}
\revi{simpleScen2_1}{where $p(d_0| d_0 \in \mathcal{A}_0)$ is the \ac{pdf} of the \ac{ue} transmitting from distance $d_0$ inside the \ac{roi}. Assuming that \ac{ue} position is uniformly distributed in $\mathcal{A}$, and letting $\Delta_0 = R_{\rm in}^2-R_{\rm min}^2$ and $\Delta_1= R_{\rm out}^2-R_{\rm in}^2$, we have $p(d_0| d_0 \in \mathcal{A}_0) = \frac{2 d_0}{\Delta_0}$ for $d_0 \in [R_{\rm min}, R_{\rm in}]$, and $p(d_0| d_0 \in \mathcal{A}_0) = 0$ otherwise.}
\revi{simpleScen2_2}{A similar expression holds for $p(d_0|d_0 \in \mathcal{A}_1)$.}Closed-form expressions of the  \ac{llr} in the two scenarios are derived in Appendix \ref{sec:llrDer}.

\revi{llrComp}{We can see that obtaining \acp{llr} needs the computation of various integrals evein in this simple case. Therefore, in general (e.g., with either multiple \acp{ap} or correlated shadowing/fading), the \ac{llr} can not be computed in closed-form, thus making  \ac{np} test problematic.}

\subsection{Estimated Distance Approach} \label{seccomp}
\revi{literature_1}{We will compare our \ac{irlv} solutions with the \ac{eda} of {\cite{li2010security}}. In \ac{eda}, first the estimate $\hat{L}(\bm{x}_{\rm UD},\bm{x}^{n}_{\rm AP})$ of the \ac{ue}-\ac{ap} distance is obtained by inverting the path-loss formula, and then the \ac{ue} position is estimated as 
}
\begin{align}
 \hat{\bm x}_{\rm UD} =  \underset{\bm x}{\arg \min} \sum_{n=1}^{N_{\rm AP}} \left(L(\bm{x},\bm{x}_{\rm AP}^{(n)}) - \hat{L}(\bm{x}_{\rm UD},\bm{x}_{\rm AP}^{(n)}) \right)^2.
 \label{probmindist}
\end{align}
\revi{literature_2}{Let $\mathcal B_0$ be the set of points of the border of $\mathcal A_0$, and let the estimated distance of the \ac{ue} from the  border $\mathcal{B}_0$ $d_{\mathcal B} = \min_{\bm{x} \in \mathcal{B}_0} \pm||\hat{\bm x}_{\rm UD} - \bm{x}||$,  where the sign is negative if $\hat{\bm x}_{\rm UD} \in \mathcal A_0$, and positive otherwise. Lastly, $d_{\mathcal B}$ is compared with a suitable threshold $d_\delta$, chosen in order to achieve a desired FA probability, resulitng in $\hat{\mathcal{H}}_{\rm MMSE} = \mathcal{H}_0, \quad \text{if } d_{\mathcal B} < d_\delta$, $\hat{\mathcal{H}}_{\rm MMSE} = \mathcal{H}_1$, otherwise.} \revi{literature_3}{Note that this approach requires the knowledge of the path-loss model (including knowledge of LOS and non-LOS state), which is quite unrealistic. Moreover, the estimator (\ref{probmindist}) is not optimal, since the position error is usually not a Gaussian variable.}

\section{\Ac{irlv} by Machine Learning Approaches}\label{sec:irlvML}

The application of the \ac{np} theorem requires the knowledge of the conditional \acp{pdf} $p(\bm{a}|\mathcal{H}_i)$ at the \acp{ap}, which can be hard to obtain also because a-priori assumptions on them may be quite unrealistic. \revi{supervised}{Therefore, we propose to use a \emph{supervised} \ac{ml} approach} operating in two phases:
\begin{itemize}
    \item {\em Learning phase}: the \acp{ap} collect attenuation vectors from a trusted \ac{ue} moving both inside and outside the \ac{roi}, while the \ac{ue} reports its position to the \acp{ap}. In this way, the \acp{ap} can learn the behaviour of the attenuation in both regions $\mathcal A_0$ and $\mathcal A_1$.
    \item {\em Exploitation phase}: the \acp{ap} verify the location of an un-trusted \ac{ue} by the attenuation's estimate, using the data acquired in the learning phase. 
\end{itemize}

The learning phase works as follows: for each training attenuation vector $\ai$,  $i=1, \ldots, S$, collected during the learning phase, there is an associated label $t_i$, $i=1, \ldots, S$, where $t_i= -1$ if the trusted \ac{ue} is in region $\mathcal{A}_0$, and $t_i = 1$ if the trusted \ac{ue} is in region $\mathcal{A}_1$. Vector $\bm{t}=[t_1, \ldots, t_S]$ collects the labels of all the  attenuation vectors in the training phase. Given these data, the \ac{ap} learns the function  $\hat{t} = f(\bm{a})\in \{-1, 1\}$ that provides the decision $\hat{\mathcal H}$ for each attenuation vector $\bm{a}$. Then, in the exploitation phase, the \ac{irlv} algorithm computes $\hat{t} = f(\bm{a})$ for a new attenuation vector and takes the decision between the two hypotheses. Note that our solution does not explicitly evaluate the \ac{pdf} and the \ac{llr}, but rather directly implements the test function.

\revi{framework2}{We stress the fact that the channel model of Section~\ref{sec:chMod} provides a realistic communication scenario, while the analysis that follows is general, as no specific channel statistics are assumed.}

In the rest of this Section, we briefly review the \ac{mlp} \ac{nn} and the \ac{svm}, describe the learning process and show that in asymptotic conditions (infinite training attenuation vectors, sufficiently complex models, and proper learning phase convergence) both \ac{mlp} and \ac{svm} functions approximate the \ac{llr} function (\ref{eq:lr}).
  
\subsection{Neural Networks}\label{sec:nn}

A \ac{nn} is a $\mathbb{R}^N \to \mathbb{R}^O$ function mapping a set of $N$ real values into $O$ real values. A \ac{nn} processes the input in $Q$ stages, named layers, where the output of one layer is the input of the next layer. Layer $0$ with (column vector) input $\bm{y}^{(0)}$ is denoted as {\em input layer}, layer $Q-1$ with (column vector) output $\bm{y}^{(Q)}$ is denoted as {\em output layer}, while intermediate layers are denoted as {\em hidden layers}. \underline{We denote as $N_L$ the number of hidden layers}.
Each layer $\ell=0, \ldots, Q-1$, has $N^{(\ell)}$ outputs obtained by processing the inputs with $N^{(\ell-1)}$ functions named neurons. The output of the $n^{\rm th}$ neuron of layer $\ell$ is
$
y_n^{(\ell+1)} = \psi^{(\ell)}\left( \bm{w}_n^{(\ell)}\bm{y}^{(\ell)}+b_n^{(\ell)} \right)$,
where the mapping between the input and the outputs is given by the {\em activation function} $\psi^{(\ell)}(\cdot)$. The argument of the activation function is a weighted linear combination, with (row vector) weights $\bm{w}_n^{(\ell)}$, of the outputs $\bm{y}^{(\ell)}$ of the previous layer plus a bias $b_n^{(\ell)}$. We focus here on feedforward \acp{nn}, i.e., without loops between neurons' input and output, an architecture also known as \ac{mlp}. For an in-depth description of \acp{nn} refer for example to \underline{\cite[Chapter 6]{goodfellow}}.\revi{hyper1}{Activation functions are typically chosen before training, while vectors $\bm{w}_n^{(\ell)}$ are adapted according to the \ac{nn} learning algorithm in order to minimize the loss function.}

In our setting, the input of the \ac{nn} is the attenuation vector $\bm{a}$, $N=N_{\rm AP}$, and the output layer has a single neuron ($O=1$) providing as output the scalar $y^{(Q)}_1$. Let $\tilde{t}(\bm{a}) = y^{(Q)}_1$ be the output of the \ac{nn} corresponding to the attenuation vector input $\bm{a}$. A threshold $\lambda$ is used on the \ac{nn} output to obtain the test function
\begin{equation}
\label{testfunNN}
    f(\bm{a}) = \begin{cases}
    1 & \tilde{t}(\bm{a}) > \lambda, \\
    -1 & \tilde{t}(\bm{a}) \leq \lambda.
    \end{cases}
\end{equation}

\revi{lambdaNN}{Parameter $\lambda$ shall be chosen in order to obtain the required \ac{fa} probability. The value of $\lambda$ which guarantees a certain \ac{fa} probability can obtained by simulation, whereas it can not be obtained by inverting the \ac{fa} probability function. This is due to the fact that the \ac{ml} framework is applied when there is no knowledge of the distribution of the variables and hence we can not compute a closed-form expression of the \ac{fa} probability}.\footnote{Notice that usually, for zero-one loss function, literature assumes $\lambda = 0.5$. However, this choice provides the control of neither \ac{fa} nor \ac{md} probabilities.}

\revi{ceNeeded}{Based on the loss function to be optimized during training \acp{nn} can solve different problems, and we consider here two widely used loss functions: \ac{mse} and \ac{ce}. }

\subsection{\ac{nn} MSE Design}
\label{sec: mse_train}
\revi{ceNeeded2}{As optimal hypothesis testing is implemented via the \ac{np} framework, which exploits the knowledge of the \ac{llr} function, we aim at learning this function from data. This problem is referred to as \emph{curve fitting} and it can be solved by training a \ac{nn} via the \ac{mse} loss function \cite{bishop92}.}
According to the \ac{mse} design criterion, the \ac{mlp} parameters are updated in the training phase in order to minimize the \ac{mse} \cite{bishop92}
\begin{equation}\label{eq:mseFunct}
\Gamma = \sum_{i=1}^S |\tilde{t}(\ai) - t_i|^2.
\end{equation}
This is achieved by using the stochastic gradient descent algorithm \underline{\cite[Section~3.1.3]{Bishop2006}}.

In order to prove the connection of \ac{nn} classifier with \ac{mse} design with the \ac{np} test, we first recall the following theorem \cite{Ruck-90}

\begin{theorem}[see \cite{Ruck-90}]
Let $g_0(\bm{a})$ be the Bayes optimal discriminant function
\begin{equation}\label{eq:bayesDisc}
g_0(\bm{a}) = \mathbb{P}(\mathcal{H}=\mathcal{H}_0|\bm{a}) - \mathbb{P}(\mathcal{H}=\mathcal{H}_1|\bm{a}).
\end{equation} 
Then the \ac{mlp} trained by backpropagation via (\ref{eq:mseFunct}) minimizes the error
$\sum_{i=1}^S \left(\tilde{t}(\bm{a}^{(i)}) - g_0(\bm{a}^{(i)})\right)^2$.
\label{teoruck}
\end{theorem}
\revi{teoruck2}{Hence, Theorem \ref{teoruck} proves that in the presence of i) perfect training, ii) an infinite number of neurons, and iii) convergence of the learning algorithm to the minimum error, the function implemented by \ac{mlp} is the Bayes optimal discriminant function. Now we have the following corollary.}
\begin{corollary}
\label{th:nn_np}
Consider an \ac{mlp} with training converged to the global minimum of the \ac{mse}, by using an infinite number of training points  ($S \rightarrow \infty$). Then the test function (\ref{testfunNN}) provides the \ac{np} test, thus it is the most powerful test.
\end{corollary}
\begin{proof}
From the Bayes rule we have 
\begin{equation}
g_0(\bm{a}) = \frac{p(\bm{a}|\mathcal H_0){\mathbb P}(\mathcal{H}=\mathcal H_0) - p(\bm{a}|\mathcal H_1){\mathbb P}(\mathcal{H}=\mathcal H_1)}{p(\bm{a}|\mathcal H_0){\mathbb P}(\mathcal{H}=\mathcal H_0) + p(\bm{a}|\mathcal H_1){\mathbb P}(\mathcal{H}=\mathcal H_1)}.
\end{equation}
Now, function (\ref{testfunNN}) imposes a threshold $\lambda$ on $g_0(\bm{a})$ and reorganizing terms we obtain $f(\bm{a}) = -1$ when
\begin{equation}
\frac{p(\bm{a}|\mathcal H_0)}{p(\bm{a}|\mathcal H_1)}> \frac{1 + \lambda}{1-\lambda} \, \frac{{\mathbb P}(\mathcal{H}=\mathcal H_1)}{{\mathbb P}(\mathcal{H}=\mathcal H_0)}  = \lambda^*,
\end{equation}
which is equivalent to the \ac{np} criterion, except for a fixed scaling of the threshold.
\end{proof}
\revi{rev11b}{Note that this result is quite general and can be applied to NNs with any number of layers and neurons, and any parameter adaptation approach, as long as the target design function is the MSE. Thus, Corollary 1 is suited also to describe the asymptotic behaviour of elaborate solutions, such as deep learning NNs.}

\subsection{\ac{nn} CE Design}
\label{sec: ce_train}

\revi{ceNeeded3}{Binary classification aims at assigning labels $0$ or $1$ to input vectors. In this case, the usual choice for the loss function is the \ac{ce} between the \ac{nn} output and the true labels of the input vector }\underline{\cite[Chapter~5.2]{Bishop2006}}
\begin{equation}\label{eq:ce}
\chi = -\sum_{i=1}^{S} t_i\ln\tilde{t}(\ai)+(1-t_i)\ln( 1-\tilde{t}(\ai)).
\end{equation}

We now prove the connection of \ac{ce} design criterion with the \ac{np} theorem.
\begin{theorem}
\label{th:nn_np2}
Consider an \ac{mlp} with training converged to the global minimum of the \ac{ce}, by using an infinite number of training points  ($S \rightarrow \infty$). Then the test function (\ref{testfunNN}) provides the \ac{np} test, thus it is the most powerful test.
\end{theorem}
\begin{proof}
The probability of being in hypothesis $\mathcal{H}_1$ given the attenuation vector $\bm{a}$ satisfies
$
    \mathbb{P}(\mathcal{H} = \mathcal{H}_1|\bm{a} ) = 1- \mathbb{P}(\mathcal{H} = \mathcal{H}_0|\bm{a} ).
$
When training is performed with the \ac{ce} loss function, the output of the \ac{mlp} is the minimum \ac{mse} approximation of the probability $\mathbb{P}(\mathcal{H}_0|\bm{a})$ of being in hypothesis $\mathcal{H}_0$, given the attenuation vector $\bm{a}$ \underline{\cite[Section~5.2]{Bishop2006}}, i.e.,
$
    \tilde{t}(\bm{a}) \approx \mathbb{P}(\mathcal{H}=\mathcal{H}_0|\bm{a})\,,
$
where the approximation is in the \ac{mse} sense. An alternative proof of this is given by \cite{nostro}.

Now, by using the threshold function (\ref{testfunNN}), we have $
    \mathbb{P}(\mathcal{H}=\mathcal{H}_0|\bm{a}) \approx  \tilde{t}(\bm{a}) > \lambda,
$ which can be rewritten as (with $\hat{\lambda}=2\lambda-1)$
\begin{equation}
    \mathbb{P}(\mathcal{H}=\mathcal{H}_0|\bm{a} )-(1-\mathbb{P}(\mathcal{H}=\mathcal{H}_0|\bm{a} )) \gtrsim \hat{\lambda}
\end{equation}
\begin{equation}
\label{lasteq}
    \mathbb{P}(\mathcal{H}=\mathcal{H}_0|\bm{a} )-\mathbb{P}(\mathcal{H}=\mathcal{H}_1|\bm{a} ) \gtrsim \hat{\lambda}.
\end{equation}
By using (\ref{testfunNN}) on the output of the \ac{nn} designed with the \ac{ce} criterion, under the convergence hypothesis, (\ref{lasteq}) coincides (except for a different threshold value) with (\ref{eq:bayesDisc}), the function implemented by the \ac{nn} trained with the \ac{mse} criterion. Therefore, from Corollary~1 we conclude that also the \ac{ce} design criterion provides a test function equivalent the \ac{np} test function.
\end{proof}

\subsection{Support Vector Machines}\label{sec:svm}

A \ac{svm} \underline{\cite[Chapter~7]{Bishop2006}} is a supervised learning model that can be used for classification and regression. We focus here on binary classification to solve the \ac{irlv} problem. The \ac{svm} implements the $\tilde{t}(\bm{a}): \mathbb{R}^{N_{\rm AP}} \to \mathbb{R}$  function 
\begin{equation}
\label{eq:svm}
\tilde{t}(\bm{a}) = \bm{w}^T \phi (\bm{a}) + b,
\end{equation}
where $\phi: \mathbb{R}^{N_{\rm AP}} \to \mathbb{R}^K$ is a feature-space transformation function, $\bm{w} \in \mathbb{R}^K$ is the weight column vector and $b$ is a bias parameter. The test function is again provided by (\ref{testfunNN}), where now $\tilde{t}(\bm{a})$ is given by (\ref{eq:svm}). Note that in the conventional \ac{svm} formulation, we have $\lambda = 0$, while here $\lambda$ is chosen according to the desired \ac{fa} probability. \revi{hyper2}{While the feature-space transformation function is chosen before training \mbox{\cite[Chapter~7]{Bishop2006}}, vector $\bm{w}$ must be properly learned from the data to obtain the desired hypothesis testing.}

We consider the \ac{lssvm} approach \cite{Suykens1999} for the optimization of the \ac{svm} parameters.  Learning for \ac{lssvm} is performed by solving the following optimization problem
\begin{subequations}
	\label{eq:lssvm}
	\begin{equation}
	\label{eq:lssvmOrig}
	\underset{\bm{w},b }{\text{min}} \quad \omega(\bm{w},b) \triangleq \frac{1}{2} \bm{w}^T \bm{w} + C \frac{1}{2} \sum_{i=1}^S e_i ^2 
	\end{equation}
	\begin{equation}
	\label{eq:stpart}
	e_i =   t_i[\bm{w}^T \phi (\bm{a}^{(i)}) + b]-1   \quad i = 1 ,\dots,S\,,
	\end{equation}
\end{subequations}
\revi{hyper3}{where $C$ is not optimized by the learning algorithm, but must be tuned on training data using a separate procedure, e.g., see \cite{guo2008novel}.}
In conventional \ac{svm}, variables $e_i$, $i=1,\ldots,S$, are constrained to be non-negative and appear in the objective function without squaring. Inequalities in the constraints translate into a quadratic programming problem, while equalities constraints in \ac{lssvm} yield a linear system of equations in the optimization values. In \cite{Yevs}, it is shown that \ac{svm} and \ac{lssvm} are equivalent under mild conditions. From  constraints  \eqref{eq:lssvm} and the fact that $t_i = \pm 1$ we have
\begin{equation}
\label{eq:els}
e_i^2 = (1 - t_i\tilde{t}(\bm{a}^{(i)}) )^2 = (t_i - \tilde{t}(\bm{a}^{(i)}))^2,
\end{equation}
that is the squared error between the soft output of the \ac{lssvm} $\tilde{t}(\bm{a}^{(i)})$ and the correct training label $t_i$.

We now prove the equivalence between the \ac{lssvm} and \ac{np} classifiers. Let us first consider the following lemma that establishes the convergence of the learning phase of \ac{svm}, as $S\rightarrow \infty$.

\begin{lemma}
	\label{lem:lem1}
	For a large number of training samples $\bm{a}^{(i)}$ taken with a given static probability distribution from a finite alphabet $\mathcal C$, i.e., for $S \rightarrow \infty$, the vector $\bm{w}$ of the \ac{lssvm} converges in probability to a vector of finite norm $||\bm{w}||_2 = \bm{w}^T\bm{w}$.
\end{lemma}

\begin{proof}
See the Appendix \ref{sec:proofTh3}.
\end{proof}
 
We can now prove the following theorem establishing the optimality of the \ac{svm} solution, as it provides the most powerful \ac{np} test for a given \ac{fa} probability.
\begin{theorem}
	\label{th:lsnp}
	\revi{revGO}{Consider a \ac{lssvm} with training converged to the global minimum of $\omega(\bm{w},b)$, and using an infinite number of training points $\bm{a}^{(i)}$ drawn from the finite alphabet $\mathcal C$.} Then the test function (\ref{testfunNN}) with (\ref{eq:svm}) provides the \ac{np} test, thus it is the most powerful test.
\end{theorem}
\begin{proof}
	From \eqref{eq:lssvmOrig} consider
	\begin{equation}
	\label{eq:lssvmDim1}
	\lim_{S \to +\infty} \frac{1}{S} \omega(\bm{w},b) =\frac{C}{2} \lim_{S \to +\infty}\frac{1}{S}  \sum_{i=1}^S e^2_i	=\frac{C}{2}\E_t(\bm{w},b),
	\end{equation}
	where $\E_t(\bm{w},b) = \Exp{e_i^2} $ is the expected value computed \wrt the training points $\bm{a}^{(i)}$, as $S$ goes to infinity. 	The first equality in \eqref{eq:lssvmDim1} comes from Lemma 1: since $\bm{w}$ converges to a finite norm, we can write
$
	\lim_{S\to \infty} \frac{1}{S} \bm{w}^T \bm{w} 	= 0.
$
	The last equality comes from the strong law of large numbers. In the limit, the optimization problem \eqref{eq:lssvm} is equivalent to
	\begin{equation}
	\label{eq:lsInf}
	\begin{aligned}
	& \underset{\bm{w},b}{\text{min}} & &  \E_t(\bm{w},b), & 
	\end{aligned}	
	\end{equation}
	where we dropped constraints  \eqref{eq:lssvm} by using \eqref{eq:els}. The optimization problem is the same as of \ac{nn} design and from \cite{Ruck-90}, with the pair $(\bm{w}^*,b^*)$ minimizing \eqref{eq:lsInf} and parametrizing \eqref{eq:svm}, we have
$
	\tilde{t}(\bm{a}^{(i)})  \approx \mathbb{P}(\mathcal{H}_0|\bm{a}^{(i)}) - \mathbb{P}(\mathcal{H}_1|\bm{a}^{(i)}).
$
	Lastly, we exploit Corollary \ref{th:nn_np} to conclude the \ac{np}-optimality of \ac{ls}-\ac{svm}.
\end{proof}

In summary, we have proven that both \ac{nn} (with \ac{ce} and \ac{mse} design) and \ac{svm} (with \ac{ls} design) converge to the \ac{np} test function as the training set size $S$ goes to infinity, thus establishing their asymptotic optimality and their relation to the theory of most powerful hypothesis testing.

\subsection{Computational Costs of \ac{ml} Approaches}
\label{sec:comp}
\revi{comp1}{In this section, we briefly review the computational cost for i) training each machine and ii) making a prediction on a new data point. Let $\eta$ be the number of epochs (how many times each training point is used) of a \ac{nn}.}

\revi{comp2}{For a basic fully connected feed-forward \ac{nn}, the backpropagation training algorithm is $\mathcal O(\eta \cdot S \cdot N_L \cdot N_{\rm AP}^3)$ when the number of neurons of each hidden layer is proportional to the input size, while the prediction of a new unseen data point is $\mathcal O(N_L \cdot N_{\rm AP}^3)$. For a more detailed analysis, which takes into account also the cost of the choice of the  activation function, see {\cite{Bianchini2014}}.}

\revi{comp3}{For a \ac{lssvm}, the estimate of the vector $\bm{w}$ at training time is found by solving a linear set of equations (instead of the traditional quadratic programming of \ac{svm}). In general, the computational cost is $\mathcal O(S^3)$; however, there are more efficient solutions that reduce this complexity to $\mathcal O(S^2)$ (see \cite{vanGestel2004}). At test time, the prediction is linear in the number of features and the number of training points, \ie $\mathcal O(N_{\rm AP} \cdot S)$.} 


\section{\ac{irlv} By One-class Classification}
\label{sec:OneClass}

\revi{revOC}{In practice, collecting training points from region ${\mathcal A}_1$ may be difficult, since this region may be large and not necessarily well defined (being simply the complement of $\mathcal A_0$).} \revi{oneClass}{Therefore, during the training phase, we collect attenuation vectors only from inside $\mathcal{A}_0$ and use them to train a \ac{ml} classifier to distinguish between vectors belonging to $\mathcal{A}_0$ and $\mathcal{A}_1$ in the testing phase.} This problem can also be denoted as one-class classification, since we have only samples taken from one of the two classes of the problem to train the models. 

In the following, we address the problem of  one-class classification  implemented via both \ac{nn} and \ac{svm}. Two approaches are considered: the \ac{ae}, using a \ac{nn}, and the \ac{oclssvm}.

Before proceeding, we consider the optimal approach when only the channel statistics from within $\mathcal A_0$ are known a-priori. In this case the \ac{llr} \eqref{eq:lr} can not be used as discriminant function, as $p(\bm{a}|\mathcal{H}_1)$ is not known. We can instead resort to the \ac{glrt} \cite{Kay-book}, which, although in general sub-optimal, is a meaningful generalization of the \ac{np} test, providing the test function 
\begin{equation}\label{eq:GLRT}
f^*(\bm{a}) =
\begin{cases}
-1 &\text{if } p(\bm{a}|\mathcal{H}_0) \geq \Lambda \\
1 & \text{if } p(\bm{a}|\mathcal{H}_0) < \Lambda.
\end{cases}
\end{equation} 

\subsection{Auto Encoder \ac{nn}}
\label{sec:auto}

\revi{deep ae}{We consider the \ac{ae} \protect{\cite{Hinton-2006}}, i.e., a \ac{nn} trained to copy the input to the output. It comprises an encoder NN (with $N_e$ layers), which transforms the $N$-dimensional input data into the $M$-dimensional code, with $M<N$, and a decoder NN (with $N_d$ layers), which reconstructs the original high-dimensional data from the low-dimensional code. For an in-depth description of the \ac{ae} architecture, please refer to \protect{\cite[Chapter 14]{goodfellow}}. 
When implementing \acp{ae}, it is convenient to use linear activation functions at the last hidden layer of the decoder \protect{\cite[Chapter 14]{goodfellow}}.} Note that the \ac{ae} output is a vector of the same size of the \ac{ae} input, and one-class classification is obtained by computing the reconstruction error between the input and the output of the \ac{ae} and comparing its absolute value with a chosen threshold.


For our \ac{irlv} problem, we train the \ac{ae} with attenuation vectors $\bm{a}^{(i)}$ taken only when the trusted \ac{ue} is in  \ac{roi} $\mathcal A_0$. Then, by letting $\bm{y}^{(Q)}(\bm a)$ be the output vector of the \ac{ae} for the attenuation input $\bm{a}$, the \ac{mse} is 
\begin{equation}\label{eq: rec err}
    \Gamma^{(AE)} = \frac{1}{N}\sum_{n=1}^{N}|a_n-y^{(Q)}_n(\bm{a})|^2.
\end{equation}
Finally, the \ac{irlv} test function  is  
\begin{equation}
f(\bm{a}) =
\begin{cases}
1 &\text{if } \Gamma^{(AE)} \geq \lambda^{(\rm AE)}, \\
-1 & \text{if } \Gamma^{(AE)} < \lambda^{(\rm AE)},
\end{cases}
\end{equation}
where again $\lambda^{(\rm AE)}$ must be chosen to achieve a desired \ac{fa} probability.

\revi{mseThresholding}{As the \ac{ae} attempts to copy the input to its output, in the testing phase only vectors with features similar to those of the training set will be reconstructed with smaller \ac{mse}, whereas input vectors with different features will be mapped to different vectors at the output, with large \ac{mse}. Since training is based on vectors collected from area $\mathcal{A}_0$ and we want to verify users located inside $\mathcal{A}_0$, by thresholding the \ac{mse}, we can obtain the desired classifier.}

About the test power of the \ac{ae}, we observe that it can be seen as the quantizer (or compression process) of an $N$-dimensional signal into an $M$-dimensional signal. In order to minimize the \ac{mse} of the reconstruction error, inputs with higher probability will have smaller quantization regions. Moreover, as the number of quantization points goes to infinity (since the quantization indices are in the continuous $M$-dimensional space) all points in the same quantization region will have approximately the same probability. However, the quantization error for points within each region will be different for each point; in particular, equal to zero for the quantization point and greater than zero at the edges of the quantization region. Thus, we can conclude that the \ac{ae} can not provide as output the \ac{pdf} of the input, even with infinite training and an infinite number of neurons, as required by the \ac{glrt} decision rule (\ref{eq:GLRT}). On the other hand, input points with a smaller \ac{pdf} belong to larger quantization regions for which the reconstruction error is {\em on average} larger; therefore, the output provided by the \ac{ae} is on average monotonically decreasing with the \ac{pdf} of the input point.

\subsection{One-Class LS-SVM}

We can also resort to \ac{svm} to perform the one-class classification in \ac{irlv}: we focus in particular on the  \ac{oclssvm}, first introduced in \cite{choi2009} as an extension of the one-class \ac{svm} \cite{Scholkopf2001estimating}. 
The only difference with respect to the \ac{svm} introduced in Section~III is that the training optimization problem is now
\begin{subequations}
	\label{eq:oneClassSvm}
	\begin{equation}
	\label{eq:oneClass1}
	\underset{\bm{w},b}{\min} \quad \omega(\bm{w}, b) \triangleq
	 \frac{1}{2} \bm{w}^T \bm{w} +  \frac{C}{2} \sum_{i=1}^S e_i^2 +b
	\end{equation}
	\begin{equation}
	\label{eq:oneClassConstr}
	\text{subject to}\, -b - \bm{w}^T \phi (\bm{a}^{(i)})  = e_i,  \quad i = 1,\dots S.
	\end{equation}
\end{subequations}
Note that in the one-class case, the bias parameter $b$ appears also in the objective function.

We observe that also for \ac{oclssvm} we can not establish a correspondence with \ac{glrt}. Nevertheless, by resorting to the Chernoff bound, we can conclude that by minimizing the \ac{mse} we also minimize the upper bound of the \ac{fa} probability; therefore, the optimization process goes in the right direction although being not optimal.

\section{Numerical Results}\label{sec:numRes}

In this section, we present the performance of the proposed \ac{irlv} methods, obtained from both experimental data and the channel models of Section~II. We consider a unitary transmitting power for each user and a carrier frequency of $f_0=2.12$~GHz, and $h_{\rm AP}^{(n)} = 15$~m for all the APs, unless differently specified. When spatial correlation of shadowing is assumed, we consider a decorrelation distance of $d_c = 75$~m according to the model of Section~II. The training points for the classification tasks are taken uniformly over the area $\mathcal A$ ($\mathcal A_0$).

For the \ac{lssvm} we use a Gaussian kernel function \protect{\cite[Chapter 6]{Bishop2006}}.
\revi{activation}{For the \ac{nn} approach we use fully connected networks. For the two-class classification problem, the activation function of the input layer is the identity function, while the activation function of neurons in the  hidden and output layers is the sigmoid \protect{\cite[Section 6.2.2.2]{goodfellow}.}} \revi{numResSimplScen3}{\acp{nn} have been trained only for \ac{ce} loss function, as we have shown in Corollary 1 and Theorem 2 that with both \ac{mse} and \ac{ce} loss functions we achieve the same performance of the \ac{np} test.}

\subsection{Two-class \ac{irlv} With Single \ac{ap}}\label{sec:singleAp}

We start with a \ac{irlv} system using a single \ac{ap}. 
\paragraph*{Uncorrelated fading/shadowing} \revi{numResSimplScen}{Firstly, we consider the environment of Section \ref{sec:los} describing a small area. The channel model includes  spatially uncorrelated fading or shadowing, with $R_{\rm out}= 10$~m, $R_{\rm in} = 2$~m, and $R_{\rm min} = 0.1$~m. Moreover, \ac{los} is assumed for path-loss. For uncorrelated fading, we consider two path-loss coefficients, namely  $\nu=2$ and $3$; the closed-form expression of the \acp{llr} for the \ac{np} test are given by (\ref{eq:llr1}) and (\ref{eq:llr2}). With spatially uncorrelated shadowing, we set $\nu=2$, and three values of shadowing standard deviation, namely  $\sigma_{s, {\rm dB}} = 0.1$~dB, $1.8$~dB, and 6~dB; the closed-form expression of the \acp{llr} for the \ac{np} test is given by (\ref{eq:llr3}). For the \ac{ml} approaches, we consider $S=10^5$ training points and a \ac{nn} with $N_L = 2$ hidden layers, each layer with $N^{(i)}=5$ neurons in layer $i = 1,2$.}

\begin{figure}
    \centering
    \includegraphics[width=8.5cm]{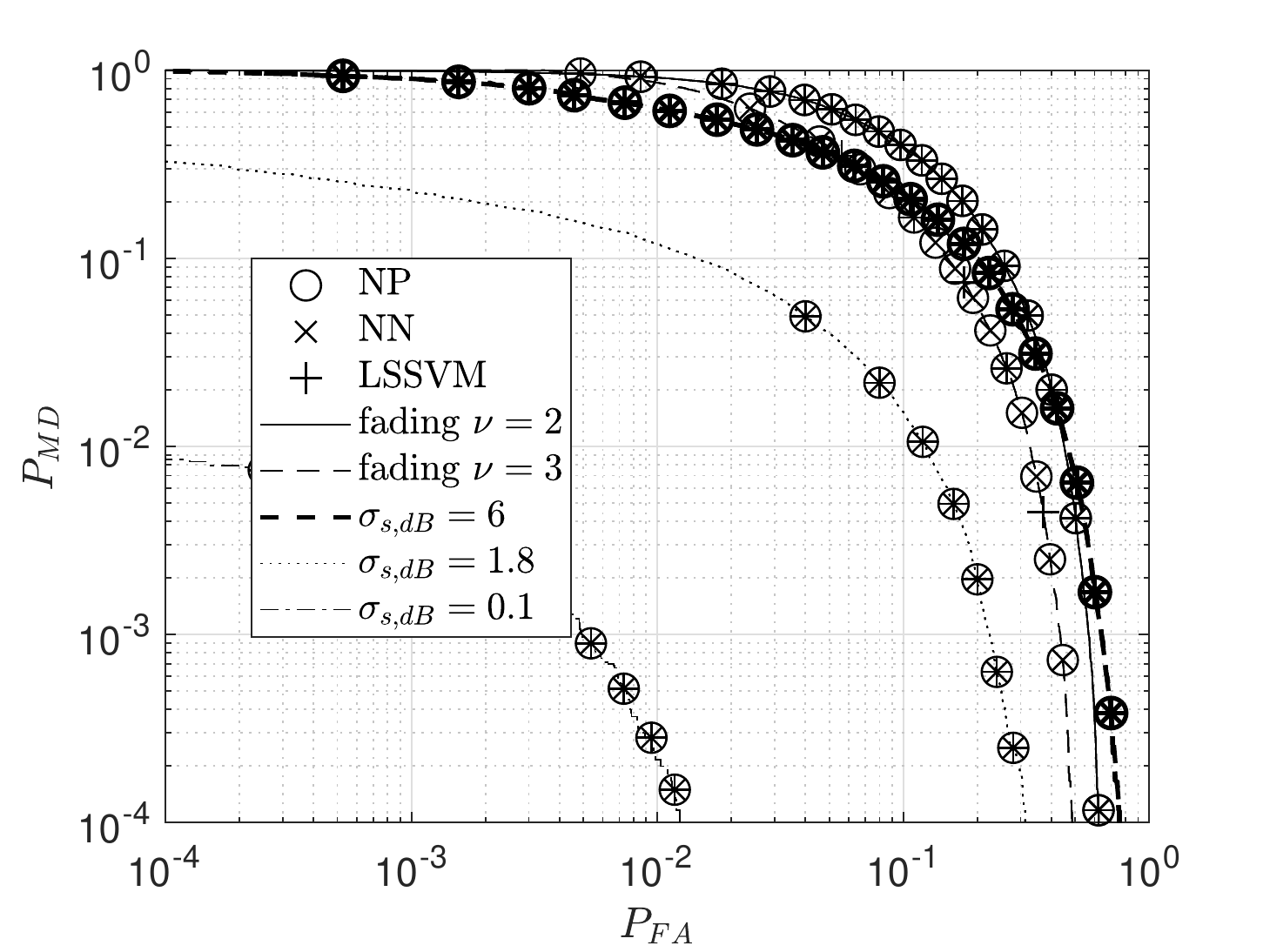}
    \caption{\acs{roc} of \acp{irlv} methods for \ac{los}, uncorrelated fading/shadowing and various values of $\nu$ and $\sigma_{s, {\rm dB}}$. Environment of Section \ref{sec:los}.}
    \label{fig:ceVSnp}
\end{figure}

\revi{numResSimplScen2}{Fig. \ref{fig:ceVSnp} shows the \ac{fa} probability versus the \ac{md} probability i.e., the  \acf{roc}, obtained with the \ac{np} test,  the \ac{nn}, and \ac{lssvm} classifiers. We notice that all models achieve the same performance, confirming our theoretical results that both \ac{nn} and \ac{lssvm} with sufficient training data and number of hidden layers are optimal as \ac{np}.} \revi{numResSimplScen4}{We observe that fading has more impact on the performance than shadowing, yielding higher \ac{fa} and \ac{md} probabilities. Still, with fading, a higher path-loss coefficient provides better results, since the attenuation increases more with the distance, thus easing classification. For spatially uncorrelated shadowing, performance improves as $\sigma_{s, {\rm dB}}$ decreases, since in this case path-loss alone already provides error-free decisions, thus the shadowing component is a disturbance in the decision process.}

\paragraph*{Spatially correlated shadowing} 

\begin{figure}[t]
    \centering
    \includegraphics[width=8.5cm]{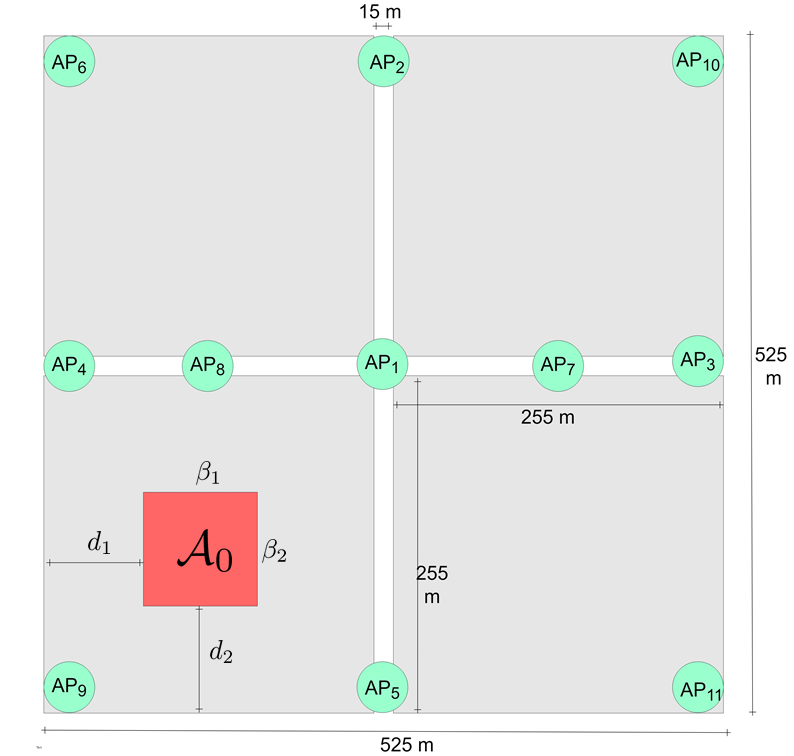}
    \caption{Reference environment.} 
    \label{fig:mBS}
\end{figure}

\begin{figure}[t]
    \centering
    \includegraphics[width=8.5cm]{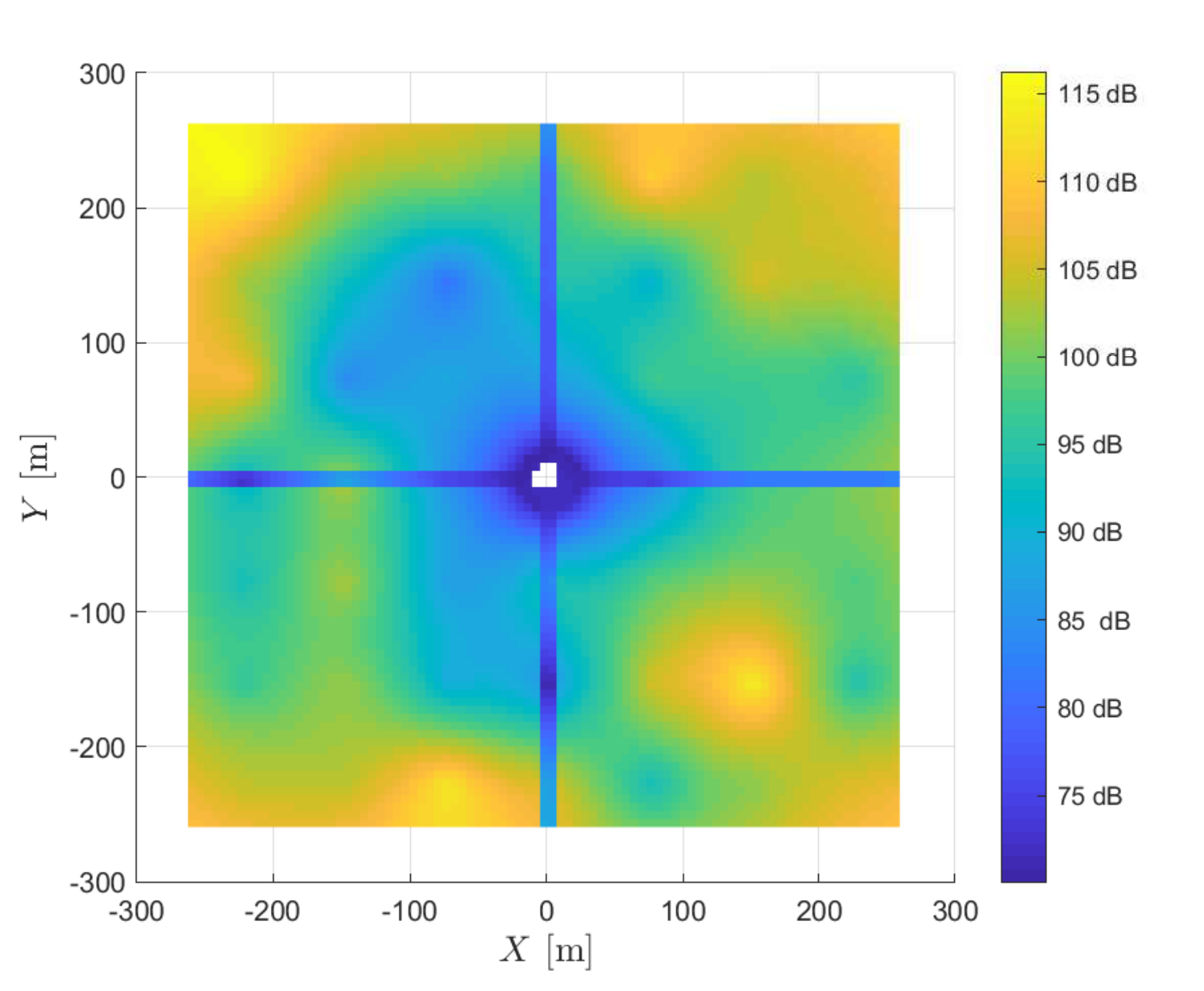}
    \caption{Example of attenuation map including path-loss and shadowing, with the \ac{ap} positioned at the center.}
    \label{fig:map}
\end{figure}

We now consider the spatially correlated shadowing  ($\sigma_{s, {\rm dB}} = 8$~dB) of Section~II. The simulation environment is shown in Fig. \ref{fig:mBS}, using only $\rm AP_1$  at the street intersection (while all other \acp{ap} are not used) and a square \ac{roi} with $d_1= 50$~m, $d_2= 50$~m, and $\beta_1 = \beta_2 = 150$~m. \revi{building}{The \ac{roi} is inside the  south-west building, modelling for example a scenario wherein privileged network resources are accessible only to users inside an office.} Along the streets, \ac{los} propagation conditions hold (with $\nu = 2$), while non-\ac{los} propagation conditions hold in the rest of the area. 
Fig. \ref{fig:map} shows a realization of the  attenuation map (including both path-loss and shadowing), highlighting the different propagation conditions. 
Since no closed-form expression of the \ac{llr} is available in this scenario, we quantize the attenuations collected in the learning phase  with a large alphabet and estimate the sampled \ac{pdf} for the quantized attenuations. Lastly, we use the estimated \ac{pdf} to compute the \acp{llr}. We use $4.46 \cdot 10^6$ training points in the area $\mathcal A$ and a uniform quantizer for the attenuation (within the observed extreme values) with $300$ quantization values. Only $10^3$ points are used for training both the \ac{mlp} and \ac{svm}.

\begin{figure}[t]
    \centering
    \includegraphics[width=8.5cm]{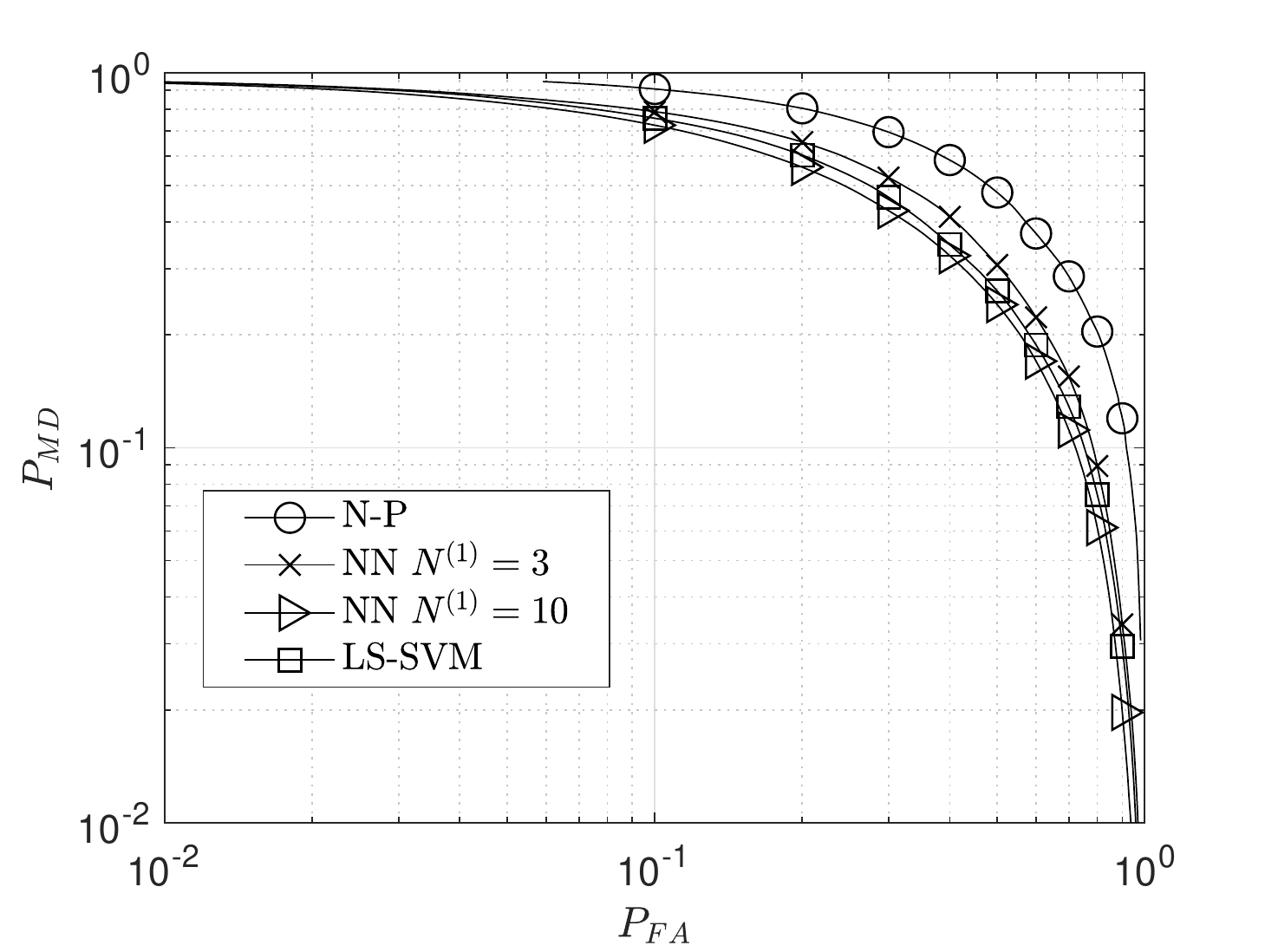}
    \caption{\ac{roc} of \ac{irlv} methods, with a \ac{nn} having $N_L=1$ and two values of $N_h$. Environment of Fig. \ref{fig:mBS}, with one \ac{ap} located at the street intersection, $d_1 = 50$~m, $d_2 = 50$~m, $\beta_1 = \beta_2 = 150$ and correlated shadowing ($\sigma_{s, {\rm dB}} = 8$~dB).}
    \label{fig:trueMap}
\end{figure}

Fig. \ref{fig:trueMap} shows the \ac{roc}  of \ac{np}, \ac{nn}, and \ac{ls}-\ac{svm}, where for a given \ac{fa} probability we report the  \ac{md} probability averaged over the shadowing attenuation maps. We notice that both \ac{nn} and \ac{ls}-\ac{svm} outperform the \ac{np} test. This means that, even for a very large number of samples available to estimate the \ac{pdf}, we still have a performance degradation with respect to \ac{np} with perfect knowledge of the statistics. On the other hand, with a small amount of training points the \ac{ml} methods outperform \ac{np}, without knowing the channel model. Therefore, in the following sections we drop the \ac{np} method.


\subsection{Two-class \ac{irlv} With Multiple \acp{ap}}
\label{sec:res_fading}

We consider the environment of Fig. \ref{fig:mBS} with $N_{\rm AP}=10$ \acp{ap} used for \ac{irlv}, namely ${\rm AP}_i$ with $i = 2,\dots,11$. The channel model includes \ac{los} and non-\ac{los} path-loss, spatially correlated shadowing ($\sigma_{\rm s, dB} = 8$ dB), and fading, as described in Section~II. We use a \ac{nn} with $L=3$ hidden layers, each layer having $N^{(i)} = 100$ neurons, $i = 1,2,3$.  

\paragraph{No fading average} We first feed the learning machine with attenuation estimates obtained without fading average, i.e., $k_f=1$. \ac{roi} position is $d_1 = 50$~m, $d_2 = 50$~m, and $\beta_1 = \beta_2 = 150$~m. Fig. \ref{fig:kf1} shows the \ac{roc} for \ac{nn}  and \ac{ls}-\ac{svm} \ac{irlv} methods and different values of the training-set size $S$. We observe that, for a given \ac{fa} probability, the average \ac{md} probability decreases as the training-set size $S$ increases. Both \ac{ml} models have similar performance with large training sets, confirming our result that they are both asymptotically optimal. However,  \ac{svm} converges faster than \ac{nn} (i.e., with a smaller $S$) to the optimal \ac{roc}. Therefore, a careful design is needed for a practical implementation with finite training and limited computational capabilities. Note that we obtain a more accurate classification with multiple \acp{ap} rather than using a single \ac{ap}. Still, for security purposes, we would prefer even lower \ac{fa} and \ac{md} probabilities; this can be achieved, for example, by increasing the number of \acp{ap} or considering other channel features, e.g., its wideband impulse response. 

\revi{revnewarea}{We have also considered a different \ac{roi} layout, with $d_1 = 100$~m, $d_2 = 255$~m and $\beta_1 = \beta_2 = 150$~m. The \ac{roi} is still positioned in the south-west corner, but it includes both the crossroads and $\rm AP_8$ (see Fig. \ref{fig:mBS}). Channel parameters are the same of Fig. \ref{fig:kf1}.}
\revi{newarea2}{Fig. \ref{fig:kf1_newArea} shows the resulting \ac{roc}, still obtained by averaging the \ac{md} probabilities over the shadowing maps. Including the street inside the \ac{roi}, with its \ac{los} path-loss, turns out to facilitate \ac{irlv} resulting in lower \ac{fa} and \ac{md} probabilities.}

\begin{figure}[t]
    \centering
    \includegraphics[width=8.5cm]{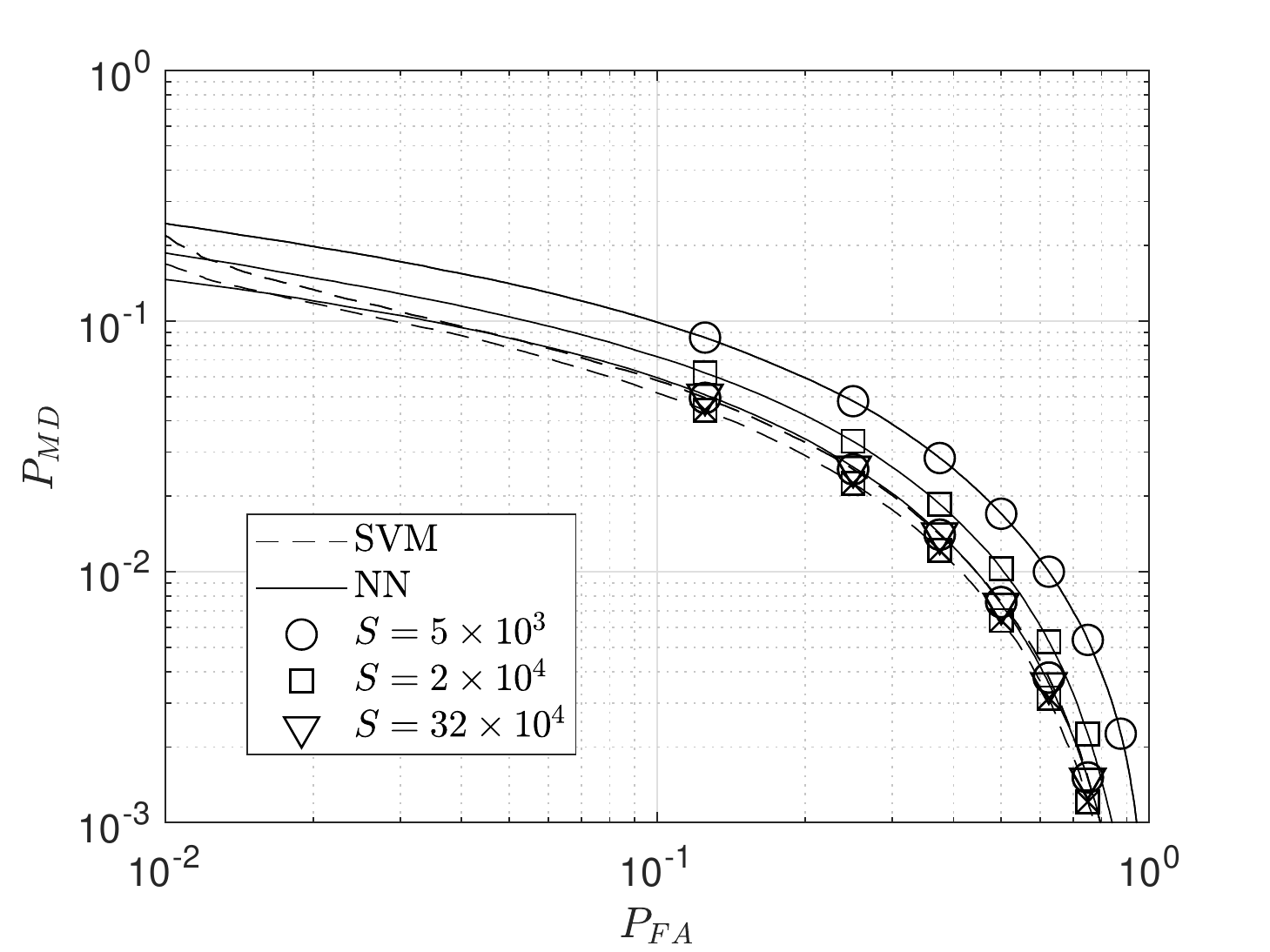}
    \caption{\ac{roc} of \ac{irlv} methods for different values of training-set size $S$. Environment of Fig. \ref{fig:mBS}, with  $N_{\rm AP}=10$, $d_1 = 50$~m, $d_2 = 50$, $\beta_1 = \beta_2 = 150$~m, and $\sigma_{s,{\rm dB}} = 8$~dB.}
    \label{fig:kf1}
\end{figure}

\begin{figure}[t]
    \centering
    \includegraphics[width=8.5cm]{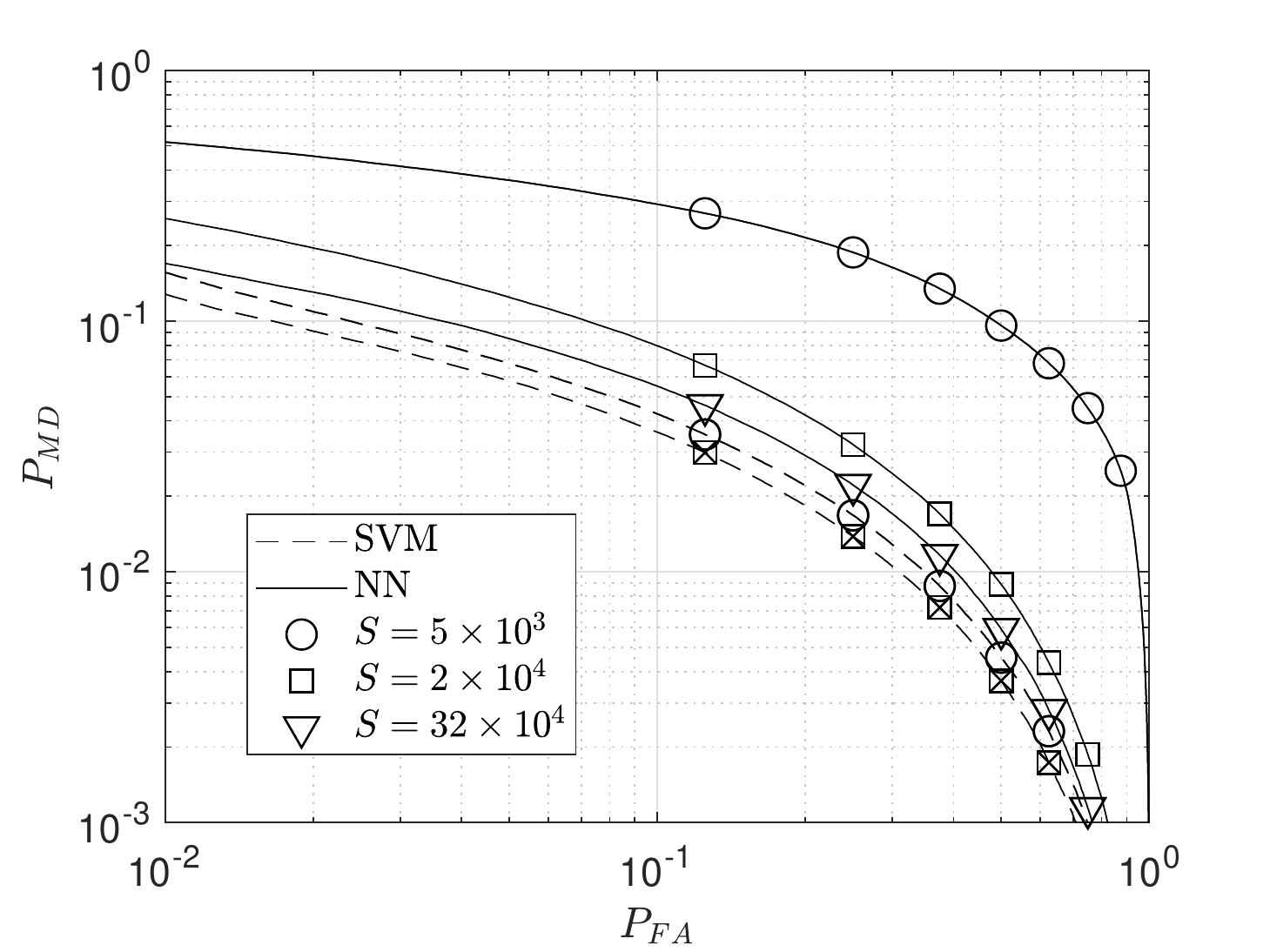}
    \caption{\ac{roc} of \ac{irlv} methods for different values of training-set size $S$. Environment of Fig. \ref{fig:mBS}, with  $N_{\rm AP}=10$, $d_1 = 100$~m, $d_2 = 225$~m, $\beta_1 = \beta_2 = 150$, and $\sigma_{s,{\rm dB}} = 8$~dB.}
    \label{fig:kf1_newArea}
\end{figure}


\begin{figure}[t]
    \centering
    \includegraphics[width=8.5cm]{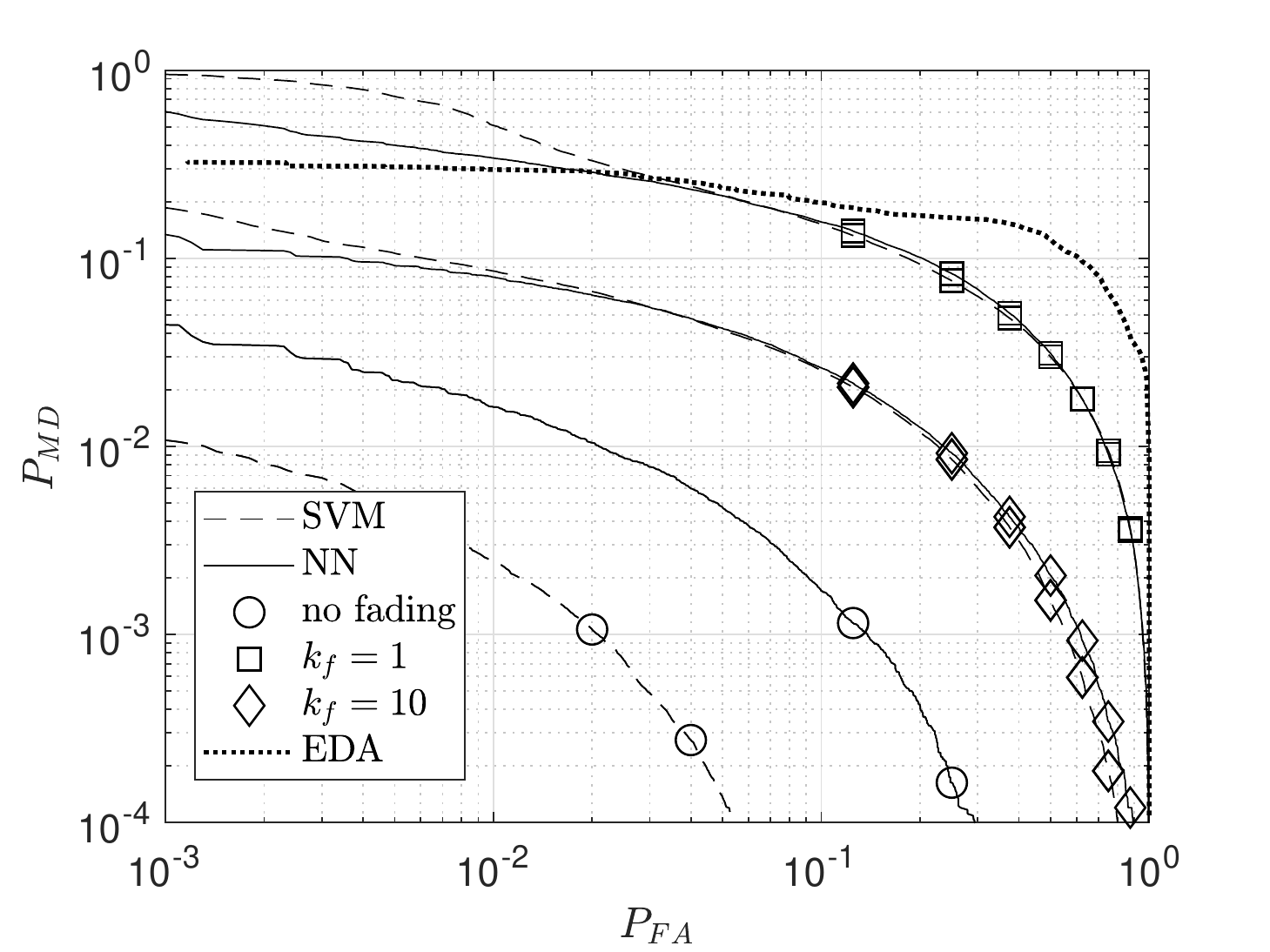}
    \caption{\ac{roc} of \ac{irlv} methods for different averages of fading. Environment of Fig. \ref{fig:mBS}, with  $N_{\rm AP}=5$, $d_1 = 50 $~m, $d_2 = 50$~m, $\beta_1 = \beta_2 = 150$~m, and $\sigma_{s,{\rm dB}} = 8$~dB.}
    \label{fig:kf10-5}
\end{figure}


\paragraph{Effect of fading average} As discussed in Section~\ref{sec:chMod}, for a given \ac{ue} position, the attenuation changes over time due to fading. By averaging $k_f$ realizations of attenuation in the same position, the effect of fading on \ac{irlv} is mitigated.
We consider the environment of Fig. \ref{fig:mBS} with $N_{AP} =5$ \acp{ap} used for \ac{irlv}, namely ${\rm AP}_i$ with $i = 1,\dots,5$, and the values of the channel parameters are those of Fig. \ref{fig:kf1}. For $n_x$  explored locations by the \ac{ue} we obtain $S= n_x \cdot k_f$ training attenuation vectors. Fig. \ref{fig:kf10-5}  shows the \ac{roc} for $k_f=1$ and 10, with $n_x = 2 \cdot 10^4$ for \ac{svm} and $n_x = 3.2 \cdot 10^5$ for \ac{nn}.  \revi{revLI2a}{We also report the performance of \ac{eda}, assuming to know the path-loss relation between the attenuation and the distance.} \revi{rev2fad}{We note that both \ac{md} and \ac{fa} probabilities can be significantly reduced by averaging fading, thus approaching the performance on channels without fading. Indeed, an average of 10 fading realizations already reduces the average \ac{md} probability from $10^{-1}$ to $10^{-2}$, for an \ac{fa} probability of $2\cdot 10^{-1}$, while we achieve an average \ac{md} probability of $4\cdot 10^{-4}$ without fading using a \ac{nn}. We also notice that, in absence of fading, \ac{svm} significantly outperforms  \ac{nn} even if \ac{nn} uses a larger $S$. This suggests that, in this scenario, the \ac{nn} has not yet converged to the optimum, wherein potentially very good performance can be achieved, due to limits in architecture, computational capabilities, and design algorithms. We should remember, in fact, that the number of parameters defining the \ac{svm} grows with the training size, while the number of parameters of the  \ac{nn} is set a-priori.} \revi{revLI2b}{Lastly, we observe that the proposed \ac{ml} techniques (both with and without fading) outperform \ac{eda}, whose performance has been obtained on channels without fading. This is due to the fact that \ac{eda} is more severely affected by shadowing seen as disturbance in the derivation of the distance, while \ac{ml} solutions may exploit it in making the decision, while still not relying on specific channel models.} 



\begin{figure}[t]
    \centering
    \includegraphics[width=8.5cm]{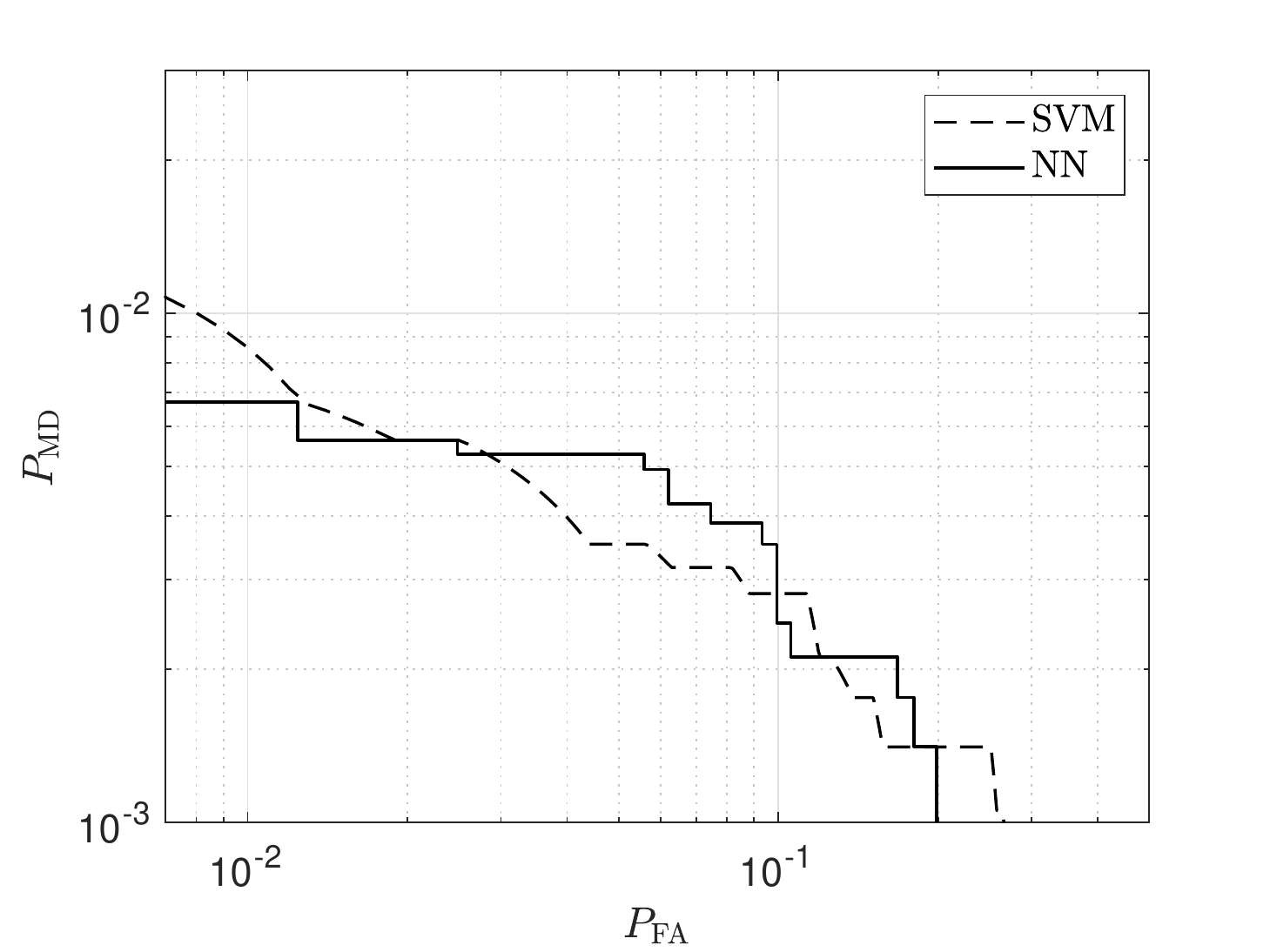}
    \caption{\ac{roc} of \ac{irlv} methods for the experimental data.}
    \label{fig:Berlinnew}
\end{figure}

\paragraph{Results on experimental data} \revi{Berlin}{We have tested the proposed \ac{irlv} solutions on real data collected by the MOMENTUM project {\cite{MOMENTUM-D53}} in a measurement campaign at Alexanderplatz in Berlin (Germany). Attenuations at the frequency of the \ac{gsm} (that may refer to a cellular IoT scenario in our \ac{irlv} context) have been measured for several \acp{ap} in an area of $4500~{\rm m} \cdot 4500~{\rm m}$, on a measurement grid of $50 \cdot 50$~m. We have considered 10 attenuation maps, corresponding to 10 \ac{ap} positions (all in meters) $\bm{x}_{\rm AP}^{(1)} = [2500, 2500]$, $\bm{x}_{\rm AP}^{(2)} = [500, 4000]$, $\bm{x}_{\rm AP}^{(3)} = [4000, 4000]$, $\bm{x}_{\rm AP}^{(4)} = [500, 500]$, $\bm{x}_{\rm AP}^{(5)} = [4000, 500]$, $\bm{x}_{\rm AP}^{(6)} = [100, 4500]$, $\bm{x}_{\rm AP}^{(7)} = [1000, 400]$, $\bm{x}_{\rm AP}^{(8)} = [4000, 500]$, $\bm{x}_{\rm AP}^{(9)} = [4300, 4000]$, and $\bm{x}_{\rm AP}^{(10)} = [4500, 500]$. The \ac{roi}  has been positioned in the lower-right corner, corresponding to, following the same notation of Fig \ref{fig:mBS}, $d_1 = 3000$~m, $ d_2 = 1500$~m, and $\beta_1 = \beta_2 = 1000 $~m.  In this case, we have a single realization of any channel effect (path-loss, shadowing, fading, \ldots) per location, for a total of 8464 realizations, 5000 of which have been used for training and the rest for testing. For \ac{nn}, we set $L = 3$ and $N^{(i)} = 500$, $i = 1,2,3$. Fig. \ref{fig:Berlinnew} shows the \ac{roc} for both \ac{nn} and \ac{lssvm}. The  performance is in line with the other figures obtained by simulation. Still, due to the small size of the available training set, \acp{roc} are not smooth. Moreover, we notice that \ac{svm} and \ac{nn} achieve approximately the same performance.}\revi{revLI}{Note also that, in order to use \ac{eda}, we should first know the path-loss to convert the attenuation estimates into distances, an information not immediately available from the experimental data. Therefore we could not compare \ac{ml} with \ac{eda} in this case, further demonstrating the utility of \ac{ml} model-less techniques for \ac{irlv}.}


\begin{figure}[t]
    \centering
    \includegraphics[width=8.5cm]{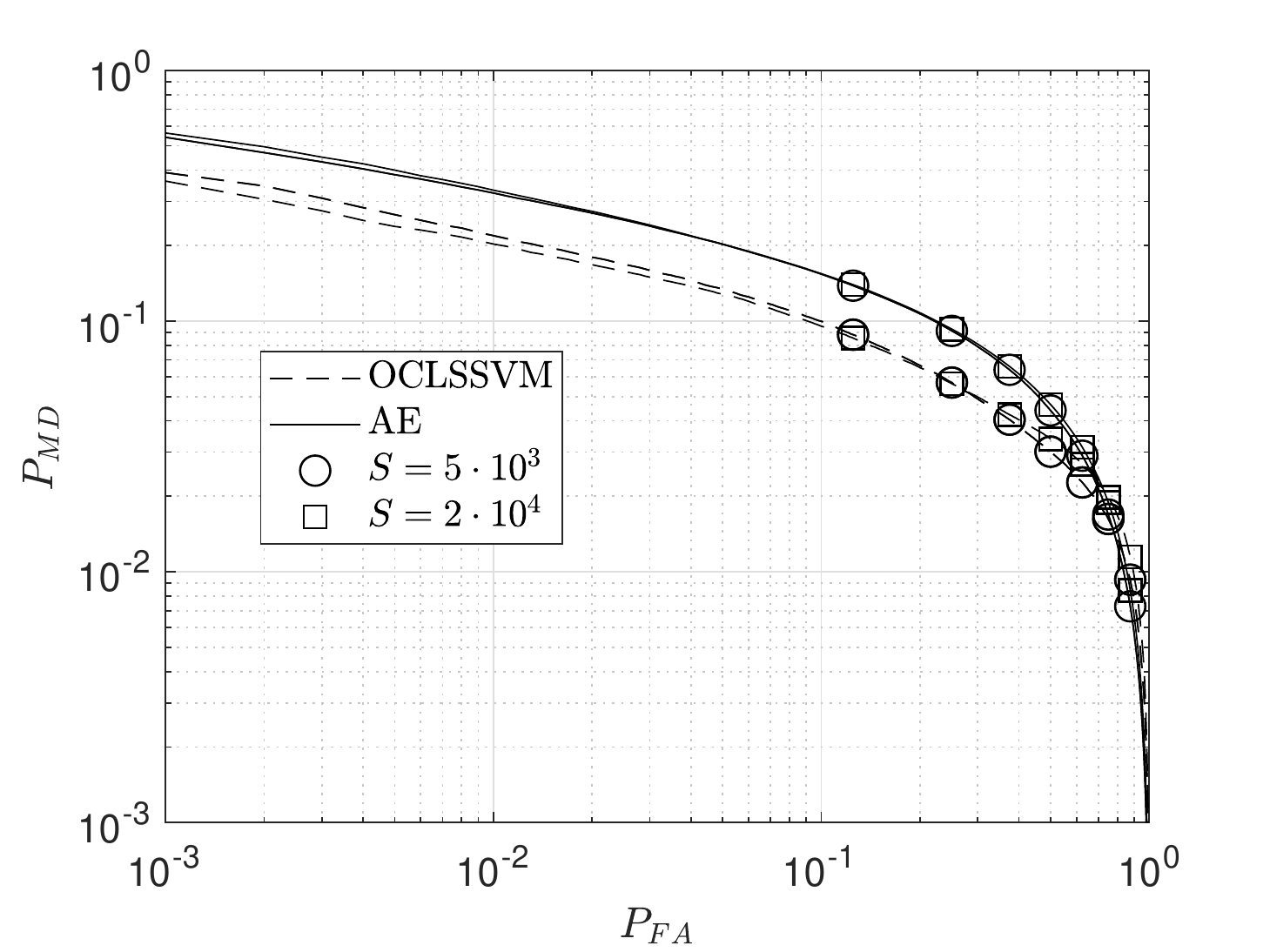}
    \caption{\ac{roc} for one-class \acp{irlv} for different training-set sizes. Environment of Fig \ref{fig:mBS} with $N_{\rm AP}=10$, $k_f=1$, and \ac{ae} with $N_L = 7$. }
    \label{fig:kf1Oc}
\end{figure}

\subsection{One-Class \ac{irlv} With Multiple \acp{ap}}\label{sec:numResOneClass}


We now focus on the one-class \ac{irlv} solutions, described in Section~\ref{sec:OneClass}, where the training points come only from the \ac{roi} $\mathcal A_0$. \revi{designAE}{The \ac{ae} has been designed according to  \cite{Hinton-2006}, i.e., all neurons use the logistic sigmoid as activation function except for those in the central hidden layer, using linear activation functions. The \ac{ae} has $N_L = 7$ hidden layers with 7, 6, 3, 2, 3, 6, and 7 neurons, respectively. Weights are initialized randomly.} The channel model is  described in Section~II, for the environment of Fig. \ref{fig:mBS} (with $N_{\rm AP} =10$), and the parameters of Section~\ref{sec:res_fading}, with $d_1=50$~m, $d_2=50$~m, and $\beta_1 = \beta_2 = 150$~m. 

 
\revi{fadingRes}{Here, we consider the effects of fading and the choice of the number of training points $S$. Fig. \ref{fig:kf1Oc} shows the \ac{roc} for one-class \ac{irlv} systems for $k_f = 1$ and two values of $S$. We first notice that both \ac{ae} and \ac{oclssvm} converge for $S = 5 \cdot 10^{3}$, and the \ac{svm}-based solution outperforms the \ac{nn}-based solution, as already seen in the case of two-class classification. Fig. \ref{fig:kf10Oc} shows the \ac{roc} for $k_f=1$ and 10, while $n_x=2 \cdot 10^4$. We note that, for both \ac{ml} techniques, averaging over fading significantly improves the performance.} \revi{revLI3}{We also report the performance of \ac{eda} obtained without fading and assuming the knowledge of the path-loss relation between attenuation and distance. Again, we note that the proposed \ac{ml} techniques significantly outperform \ac{eda} (in the absence of fading).} In the figure we also report the performance of two-class \ac{svm} for channels without fading: we can observe that, in the considered scenario, two-class \ac{irlv} outperforms the one-class \ac{irlv}: the former achieves a lower $P_{\rm MD}$ for the same $P_{\rm FA}$. This result is expected since the two-class \ac{irlv} also exploits the (estimated) statistics of attenuation while under attacks.

\begin{figure}[t]
    \centering
    \includegraphics[width=8.5cm]{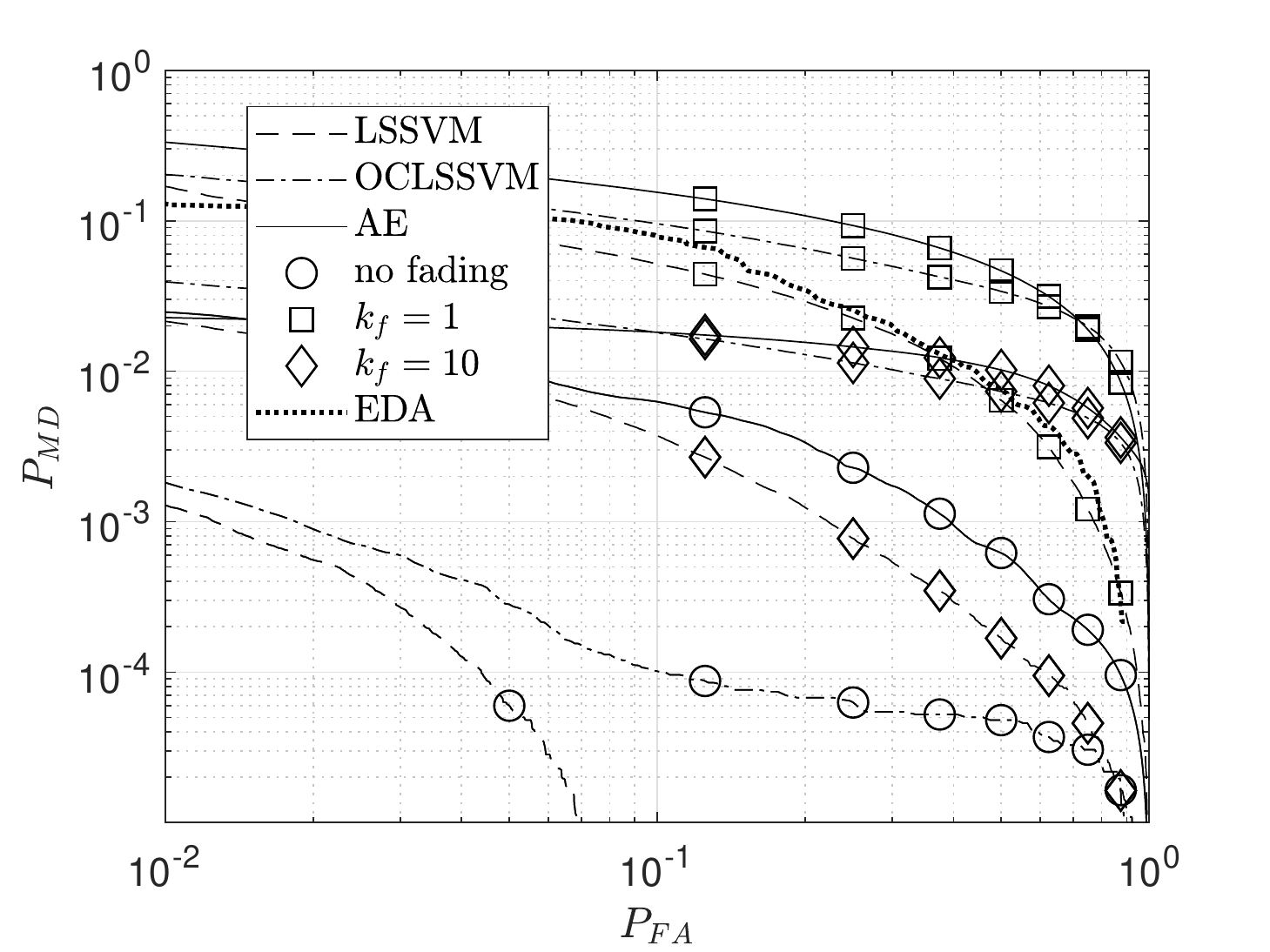}
    \caption{\ac{roc} of one-class \acp{irlv} for different values of $k_f$. Environment of Fig. \ref{fig:mBS} with $N_{\rm AP}=10$, and $n_x= 2 \cdot 10^4$,  \ac{ae} with $N_L = 7$.  }
    \label{fig:kf10Oc}
\end{figure}

\section{Conclusions}

In this paper, we have proposed innovative solutions for \ac{irlv} in wireless networks that exploit the features of the channels between the \ac{ue} whose location must be verified by a trusted network of \acp{ap}. By observing that in typical situations the channel statistics are not available for \ac{irlv}, we have proposed \ac{ml}-based solutions, operating with both one- and two-class classification, i.e., with and without a-priori assumptions on attack statistics. For two-class classification we have proved that  both \ac{nn} and \ac{svm} solutions  are the most powerful tests for a given sensitivity, i.e., they are equivalent to the \ac{np} test. Instead, for one-class classification both \ac{ae} and \ac{svm} solutions are not equivalent to the \ac{glrt}. We have also investigated how to collect the training points in order to be robust against the channel fading.

\appendices

\section{LLRs derivation}
\label{sec:llrDer}

\subsubsection{Uncorrelated Fading scenario}
\revi{simpleScen3}{Assuming spatially uncorrelated Rayleigh fading, without shadowing (\ie $\sigma_{s,\rm dB}=0)$, given a \ac{ue} located at distance $d$, the channel gain $g=1/a$ is exponentially distributed with mean (in dB) $P_{\rm PL,LOS}(d)$ given by $\eqref{eq:los}$. Letting}
\begin{equation}
    F(\Delta,R_0,R_1) = 
    \frac{2}{\Delta}\int_{R_0}^{R_1} 10^{P_{{\rm PL},{\rm LOS}}(d_0)/10} \exp\left(-10^{P_{{\rm PL},{\rm LOS}}(d_0)/10} \frac{1}{a}\right)d_0 \, {\tt d}d_0,
\end{equation}
\revi{simpleScen3_1}{from the uniform \ac{ue} distribution and \eqref{eq:prc} we have $p(a|\mathcal{H}_0)=F(\Delta_0,R_{\rm min},R_{\rm in})$, whereas  $p(a|\mathcal{H}_1) = F(\Delta_1,R_{\rm in},R_{\rm out})$.}
\revi{simpleScen3_2}{By computing integrals for path-loss coefficient $\nu = 2$, the \ac{llr} is} 
\begin{equation}\label{eq:llr1}
   \mathcal{M}(a) =
   \ln\left(\frac{R^2-R_{\rm min}^2}{R_{\rm in}^2-R_{\rm in}^2}\frac{\mathcal{V}(R_{\rm min},a)-\mathcal{V}(R_{\rm in},a)}{\mathcal{V}(R_{\rm in},a)-\mathcal{V}(R,a)}\right),
\end{equation}
 \begin{equation}
\mathcal{V}(d_0,a) = \exp\left(-\frac{1}{a}\left(\frac{4 \pi f_0 d_0}{c}\right)^2\right) \left(\frac{1}{a}\left(\frac{4 \pi f_0 d_0}{c}\right)^2+1\right) .   
\end{equation}

\revi{simpleScen3_3}{Let $\Gamma(\gamma,b)= \int_{b}^{\infty}t^{\gamma-1}e^{-t} dt$ be the incomplete gamma function, then for $\nu=3$ we have instead}
\begin{equation}\label{eq:llr2}
 \mathcal{M}(a) =
 \ln\left(\frac{R^2-R_{\rm in}^2}{R_{\rm in}^2-R_{\rm min}^2}\frac{\Gamma\left(\frac{5}{3},\frac{1}{a}\left(\frac{4 \pi f_0}{c}\right)^3 R_{\rm min}^3\right)-\Gamma\left(\frac{5}{3},\frac{1}{a}\left(\frac{4 \pi f_0}{c}\right)^3 R_{\rm in}^3\right)}{\Gamma\left(\frac{5}{3},\frac{1}{a}\left(\frac{4 \pi f_0}{c}\right)^3 R_{\rm in}^3\right)-\Gamma\left(\frac{5}{3},\frac{1}{a}\left(\frac{4 \pi f_0}{c}\right)^3 R^3\right)}\right),
\end{equation}
 
\subsubsection{Uncorrelated shadowing scenario}
\revi{simpleScen4}{Assuming spatially uncorrelated shadowing, without fading   we have $10\log_{10}a^{(n)}=P^{(n)}_{\rm PL}+s$, \ie the received power from a given location is distributed in the logarithmic domain as a Gaussian random variable with mean value given by the path-loss (\ref{eq:los}) and standard deviation $\sigma_{s,\rm dB}$. Letting}
\begin{equation}\label{eq:sh}
   G(\Delta,R_0,R_1) = 
    \frac{2}{\Delta}\int_{R_0}^{R_1} \exp\left(-\frac{1}{2}\frac{\left(\frac{1}{a}+10\nu\log_{10}\left(\frac{4 \pi f_0 d_0}{c}\right)\right)^2}{\sigma_{s,\rm dB}^2}\right) d_0 \, {\tt d} d_0, 
\end{equation}
\revi{simpleScen4_1}{from (\ref{eq:prc}), the \ac{pdf} of incurring an attenuation $a$ in hypothesis $\mathcal{H}_0$ is $p(a|\mathcal{H}_0)=G(\Delta_0,R_{\rm min}, R_{\rm in})$, and $p(a|\mathcal{H}_1) = G(\Delta_1, R_{\rm in},R_{\rm out})$.}
\revi{simpleScen4_2}{By solving the integral in (\ref{eq:sh}) we obtain the \ac{llr}}
\begin{equation}\label{eq:llr3}
    \mathcal{M}(a) = \ln\left(\frac{R_{\rm out}^2}{R_{\rm in}^2} \frac{\mathcal{T}(R_{\rm in})-\mathcal{T}(R_{\rm min})}{\mathcal{T}(R_{\rm out})-\mathcal{T}(R_{\rm in})}\right),
\end{equation}
\revi{simpleScen4_3}{where $\erf(x)= \frac{2}{\sqrt{\pi}}\int_0^x e^{-t^2} {\tt d}t$ is the error function and}
\begin{equation}
    \mathcal{T}(d_0) = \erf\left( \frac{\frac{100 \nu^2}{\sigma_{s,\rm dB}^2}\ln d_0-\ln^2(10)+\frac{\frac{1}{a} 10 \nu \ln 10}{2\sigma_{s,\rm dB}^2}}{\sqrt{1/2\sigma_{s,\rm dB}^2}10\nu\ln 10}\right).
\end{equation}

\section{Proof of Theorem 3}\label{sec:proofTh3}

	Given a finite  attenuation vector alphabet $\mathcal C = \{\bm{\alpha}_1, \ldots, \bm{\alpha}_M\}$ of $M$ elements, with $\bm{a}^{(i)} \in \mathcal C$, we indicate with $p_{\bm{a}^{(i)},t_i}(\bm{\alpha}_j, t)$, with $t \in \{-1,1\}$, the joint probability of input vector $\bm{a}^{(i)}$ and corresponding output $t_i$, $i=1, \ldots, S$.
	
	By the Glivenko–Cantelli theorem we have that with probability 1 as $S\rightarrow \infty$ there are $Sp_{\bm{a}^{(i)},t_i}(\bm{\alpha}_j,t)$ training vectors $\bm{\alpha}_j$ with associated true label $t$ in any training sequence.
	All these training points will have the same error values $\epsilon_j$, from (\ref{eq:stpart}), that will appear $Sp_{\bm{a}^{(i)},t_i}(\bm{\alpha}_j,t)$ times in the sum $\sum_{i=1}^{S} e_i^2$.
	Note that in the training ensemble there could be two equal instances $\bm{a}^{(m)}=\bm{a}^{(n)}=\bm{\alpha}_j$, but with different labels $t_m \neq t_n$. Therefore, for a given $\bm{\alpha}_j$ we can have two possible errors, depending on $t_i$, and we denote them with $\epsilon_{j,1}$ and $\epsilon_{j,-1}$.
	This translates into only $2M$ \textit{distinct} constraints of type \eqref{eq:stpart}.
	Asymptotically, for $S \to \infty$, problem (\ref{eq:lssvm}) becomes
		\begin{equation}
		\label{eq:lssvm2}
		\underset{\bm{w},e}{\text{min}} \quad f_l' \triangleq \frac{1}{2} \bm{w}^T \bm{w} + C S \frac{1}{2} \sum_{j=1}^M [p_{\bm{a}^{(i)},t_i}(\bm{\alpha}_j,1) \epsilon_{j,1}^2 + p_{\bm{a}^{(i)},t_i}(\bm{\alpha}_j,-1) \epsilon_{j,-1}^2]  
		\end{equation}
		subject to 
$
		[\bm{w}^T \phi (\bm{\alpha}_j) + b] = 1- \epsilon_{j,1}
$
and
$
-[\bm{w}^T \phi (\bm{\alpha}_j) + b] = 1- \epsilon_{j,-1}\quad j = 1 ,\dots,M,
$
	whose solution provides the convergence value (in probability) of vector $\bm{w}$. We write the Lagrangian
	\begin{equation}
	\mathcal{L}_1 = f_l' - \sum_{j=1}^{M} v_j \left[ \bm{w}^T \phi (\bm{\alpha}_j) + b - 1 + \epsilon_{j,1} \right] 
	- \sum_{j=1}^{M} u_j \left[- \bm{w}^T  \phi (\bm{\alpha}_j) - b  - 1 + \epsilon_{j,-1} \right], 
	\end{equation}
	where $\{u_k,v_k\}_{k=1}^{M}$ are the Lagrangian multipliers. By setting to zero the derivatives with respect to $\{\bm{w},b,\epsilon_{j,1},\epsilon_{j,-1}, v_j,u_j\}$  we get the system of equations
	\begin{subequations}
		\label{eq:system1}
		\begin{equation}
		\sum_{k=1}^{M} (u_k - v_k) k(\phi (\bm{\alpha}_k,\bm{\alpha}_j)) + b - 1 + \frac{v_j}{CSp_{\bm{a}^{(i)},t_i}(\bm{\alpha}_j,1)} = 0
		\quad j=1\dots M,
		\end{equation}
		\begin{equation}
		- \sum_{k=1}^{M} (u_k - v_k) k(\phi (\bm{\alpha}_k,\bm{\alpha}_j)) - b - 1 + \frac{v_j}{CSp_{\bm{a}^{(i)},t_i}(\bm{\alpha}_j,-1)} = 0
		\quad j=1,\dots, M,
		\end{equation}
		\begin{equation}
		\sum_{k=1}^{M} (u_k - v_k) = 0.
		\end{equation}
	\end{subequations}
	Note that \eqref{eq:system1} is a system with $2M + 1$ equations, linear in the $2M + 1$ unknowns $\{u_k,v_k,b\}_{k=1}^{k=M}$ and therefore has finite solution. In particular, we have
	\begin{equation}
	\label{eq:wSolution}
	\bm{w}^T\bm{w} =  \sum_{k=1}^{M} \sum_{h=1}^{M} k(\bm{\alpha}_k,\bm{\alpha}_h) (v_kv_h + u_ku_h -2 v_ku_h),
	\end{equation}
	where we used the fact that the kernel function
$
	k(\bm{\alpha}_k,\bm{\alpha}_h) \triangleq \phi(\bm{\alpha}_k) \phi(\bm{\alpha}_h)^T
$
	 is symmetric \wrt its inputs. 	We conclude that $\bm{w}$ has a finite norm since the right hand side of \eqref{eq:wSolution} is a finite sum.


\bibliographystyle{IEEEtran}
\bibliography{bibliography}

\begin{thebibliography}{10}
\providecommand{\url}[1]{#1}
\csname url@samestyle\endcsname
\providecommand{\newblock}{\relax}
\providecommand{\bibinfo}[2]{#2}
\providecommand{\BIBentrySTDinterwordspacing}{\spaceskip=0pt\relax}
\providecommand{\BIBentryALTinterwordstretchfactor}{4}
\providecommand{\BIBentryALTinterwordspacing}{\spaceskip=\fontdimen2\font plus
\BIBentryALTinterwordstretchfactor\fontdimen3\font minus
  \fontdimen4\font\relax}
\providecommand{\BIBforeignlanguage}[2]{{%
\expandafter\ifx\csname l@#1\endcsname\relax
\typeout{** WARNING: IEEEtran.bst: No hyphenation pattern has been}%
\typeout{** loaded for the language `#1'. Using the pattern for}%
\typeout{** the default language instead.}%
\else
\language=\csname l@#1\endcsname
\fi
#2}}
\providecommand{\BIBdecl}{\relax}
\BIBdecl

\bibitem{Zeng-survey}
Y.~Zeng, J.~Cao, J.~Hong, S.~Zhang, and L.~Xie, ``Secure localization and
  location verification in wireless sensor networks: a survey,'' \emph{The
  Jour. of Supercomputing}, vol.~64, no.~3, pp. 685--701, 11 2013.

\bibitem{8376254}
G.~Caparra, M.~Centenaro, N.~Laurenti, and S.~Tomasin, ``Optimization of anchor
  nodes' usage for location verification systems,'' in \emph{Proc. 2017 Int.
  Conf. on Localization and GNSS (ICL-GNSS), {N}ottingham, {UK}}, 6 2017, pp.
  1--6.

\bibitem{wei2013}
Y.~Wei and Y.~Guan, ``Lightweight location verification algorithms for wireless
  sensor networks,'' \emph{IEEE Trans. Parallel and Distributed Systems},
  vol.~24, no.~5, pp. 938--950, 5 2013.

\bibitem{7903611}
L.~C. {\it et al.}, ``Robustness, security and privacy in location-based
  services for future {IoT}: A survey,'' \emph{IEEE Access}, vol.~5, pp.
  8956--8977, 4 2017.

\bibitem{quaglia}
E.~A. Quaglia and S.~Tomasin, ``Geo-specific encryption through implicitly
  authenticated location for {5G} wireless systems,'' in \emph{Proc. 2016 IEEE
  17th Int. Workshop on Signal Processing Advances in Wireless Commun. (SPAWC),
  {E}dinburgh, {UK}}, 7 2016, pp. 1--6.

\bibitem{li2010security}
C.~Li, F.~Chen, Y.~Zhan, and L.~Wang, ``Security verification of location
  estimate in wireless sensor networks,'' in \emph{Proc. Int. Conf. on Wireless
  Commun. Networking and Mobile Computing (WiCOM)}, 9 2010, pp. 1--4.

\bibitem{7270404}
E.~Jorswieck, S.~Tomasin, and A.~Sezgin, ``Broadcasting into the uncertainty:
  Authentication and confidentiality by physical-layer processing,''
  \emph{Proceedings of the IEEE}, vol. 103, no.~10, pp. 1702--1724, 10 2015.

\bibitem{Baracca-12}
P.~Baracca, N.~Laurenti, and S.~Tomasin, ``Physical layer authentication over
  {MIMO} fading wiretap channels,'' \emph{IEEE Trans. Wireless Commun.},
  vol.~11, no.~7, pp. 2564--2573, 5 2012.

\bibitem{7398138}
L.~{Xiao}, Y.~{Li}, G.~{Han}, G.~{Liu}, and W.~{Zhuang}, ``{PHY}-layer spoofing
  detection with reinforcement learning in wireless networks,'' \emph{IEEE
  Trans. Vehic. Tech.}, vol.~65, no.~12, pp. 10\,037--10\,047, 12 2016.

\bibitem{Brands}
S.~Brands and D.~Chaum, ``Distance-bounding protocols,'' in \emph{Proc.
  Workshop on the Theory and Application of Cryptographic Techniques}.\hskip
  1em plus 0.5em minus 0.4em\relax Lofthus, Norway: Springer, 7 1993, pp.
  344--359.

\bibitem{singelee2005location}
D.~Singelee and B.~Preneel, ``Location verification using secure distance
  bounding protocols,'' in \emph{Proc. IEEE Int. Conf. on Mobile Adhoc and
  Sensor Systems, 2005.}, 12 2005, p.~7.

\bibitem{song2008secure}
J.-H. Song, V.~W. Wong, and V.~C. Leung, ``Secure location verification for
  vehicular ad-hoc networks,'' in \emph{Proc. IEEE Global Telecommunications
  Conf.}, 12 2008, pp. 1--5.

\bibitem{Sastry}
N.~Sastry, U.~Shankar, and D.~Wagner, ``Secure verification of location
  claims,'' in \emph{Proc. of the 2nd ACM workshop on Wireless security}.\hskip
  1em plus 0.5em minus 0.4em\relax San Diego (CA): ACM, 9 2003, pp. 1--10.

\bibitem{Vora}
A.~Vora and M.~Nesterenko, ``Secure location verification using radio
  broadcast,'' \emph{IEEE Trans. Dependable and Secure Computing}, vol.~3,
  no.~4, pp. 377--385, 10 2006.

\bibitem{7145434}
A.~Abdou, A.~Matrawy, and P.~C. van Oorschot, ``{CPV}: Delay-based location
  verification for the {I}nternet,'' \emph{IEEE Trans. Dependable and Secure
  Computing}, vol.~14, no.~2, pp. 130--144, 3 2017.

\bibitem{yan2016location}
S.~Yan, I.~Nevat, G.~W. Peters, and R.~Malaney, ``Location verification systems
  under spatially correlated shadowing,'' \emph{IEEE Trans. Wireless Commun.},
  vol.~15, no.~6, pp. 4132--4144, 2 2016.

\bibitem{Cover-book}
T.~M. Cover and J.~A. Thomas, \emph{Elements of information theory}.\hskip 1em
  plus 0.5em minus 0.4em\relax John Wiley \& Sons, 2012.

\bibitem{xiao-2018}
L.~Xiao, X.~Wan, and Z.~Han, ``Phy-layer authentication with multiple landmarks
  with reduced overhead,'' \emph{IEEE Trans. Wireless Commun.}, vol.~17, no.~3,
  pp. 1676--1687, 12 2018.

\bibitem{tian2015robust}
Y.~Tian, B.~Denby, I.~Ahriz, P.~Roussel, and G.~Dreyfus, ``Robust indoor
  localization and tracking using {GSM} fingerprints,'' \emph{EURASIP Jour. on
  Wireless Commun. and Networking}, vol. 2015, no.~1, p. 157, 6 2015.

\bibitem{3gpp}
``{LTE}; evolved universal terrestrial radio access ({E}-{UTRA}); radio
  frequency ({RF}) system scenarios,'' 3GPP, TR 36.942 version 15.0.0 Release
  15, Jul 2018.

\bibitem{goldsmith2005}
A.~Goldsmith, \emph{Wireless communications}.\hskip 1em plus 0.5em minus
  0.4em\relax Cambridge university press, Cambridge, 2005.

\bibitem{Kay-book}
S.~M. Kay, \emph{Fundamentals of Statistical Signal Processing: Estimation
  Theory}.\hskip 1em plus 0.5em minus 0.4em\relax Upper Saddle River, NJ, USA:
  Prentice-Hall, Inc., 1993.

\bibitem{goodfellow}
I.~Goodfellow, Y.~Bengio, and A.~Courville, \emph{Deep learning}.\hskip 1em
  plus 0.5em minus 0.4em\relax MIT press, Cambridge, 2016.

\bibitem{bishop92}
C.~M. Bishop and C.~Roach, ``Fast curve fitting using neural networks,''
  \emph{Review of scientific instruments}, vol.~63, no.~10, pp. 4450--4456, 6
  1992.

\bibitem{Bishop2006}
C.~M. Bishop, \emph{{Pattern Recognition And Machine Learning}}.\hskip 1em plus
  0.5em minus 0.4em\relax Springer, New York, 2006.

\bibitem{Ruck-90}
D.~W. Ruck, S.~K. Rogers, M.~Kabrisky, M.~E. Oxley, and B.~W. Suter, ``The
  multilayer perceptron as an approximation to a {B}ayes optimal discriminant
  function,'' \emph{IEEE Trans. Neural Networks}, vol.~1, no.~4, pp. 296--298,
  12 1990.

\bibitem{nostro}
A.~{Brighente}, F.~{Formaggio}, M.~{Centenaro}, G.~M. {Di Nunzio}, and
  S.~{Tomasin}, ``Location-verification and network planning via machine
  learning approaches,'' \emph{arXiv e-prints}, 11 2018.

\bibitem{Suykens1999}
J.~A. Suykens and J.~Vandewalle, ``Least squares support vector machine
  classifiers,'' \emph{Neural processing letters}, vol.~9, no.~3, pp. 293--300,
  6 1999.

\bibitem{guo2008novel}
X.~Guo, J.~Yang, C.~Wu, C.~Wang, and Y.~Liang, ``A novel {LS-SVMs}
  hyper-parameter selection based on particle swarm optimization,''
  \emph{Neurocomputing}, vol.~71, no. 16-18, pp. 3211--3215, 10 2008.

\bibitem{Yevs}
J.~Ye and T.~Xiong, ``{SVM} versus least squares {SVM},'' \emph{Artificial
  Intelligence and Statistics}, pp. 644--651, 4 2007.

\bibitem{Bianchini2014}
M.~Bianchini and F.~Scarselli, ``On the complexity of neural network
  classifiers: A comparison between shallow and deep architectures,''
  \emph{IEEE Trans. Neural Networks and Learning Systems}, vol.~25, no.~8, pp.
  1553--1565, 8 2014.

\bibitem{vanGestel2004}
T.~van Gestel, J.~A. Suykens, B.~Baesens, S.~Viaene, J.~Vanthienen, G.~Dedene,
  B.~de~Moor, and J.~Vandewalle, ``Benchmarking least squares support vector
  machine classifiers,'' \emph{Machine Learning}, vol.~54, no.~1, pp. 5--32, 1
  2004.

\bibitem{Hinton-2006}
G.~E. Hinton and R.~R. Salakhutdinov, ``Reducing the dimensionality of data
  with neural networks,'' \emph{Science}, vol. 313, no. 5786, pp. 504--507, 7
  2006.

\bibitem{choi2009}
Y.-S. Choi, ``Least squares one-class support vector machine,'' \emph{Pattern
  Recognition Letters}, vol.~30, no.~13, pp. 1236--1240, 10 2009.

\bibitem{Scholkopf2001estimating}
B.~Sch{\"o}lkopf, J.~C. Platt, J.~Shawe-Taylor, A.~J. Smola, and R.~C.
  Williamson, ``Estimating the support of a high-dimensional distribution,''
  \emph{Neural computation}, vol.~13, no.~7, pp. 1443--1471, 7 2001.

\bibitem{MOMENTUM-D53}
\BIBentryALTinterwordspacing
H.~G. {\it et al.}, ``Evaluation of reference and public scenarios,''
  IST-2000-28088 MOMENTUM, Tech. Rep. D5.3, 2003. [Online]. Available:
  \url{http://www.zib.de/momentum/paper/momentum-d53.pdf}
\BIBentrySTDinterwordspacing

\end{thebibliography}
\end{document}